\documentclass[journal,onecolumn,10pt]{IEEEtran}
\IEEEoverridecommandlockouts

\usepackage{cite}
\usepackage{amsmath,amssymb,amsfonts}
\usepackage[noend]{algpseudocode}
\usepackage{amsthm}

\usepackage{graphicx}
\usepackage{float}
\usepackage{threeparttable}
\usepackage[ruled,linesnumbered]{algorithm2e}
\usepackage{algpseudocode}
\usepackage{textcomp}
\usepackage{xcolor}
 \usepackage{makecell}
\usepackage{pdfpages}
\usepackage{footnote}
\usepackage{diagbox}
\usepackage{bbm}
\usepackage{dsfont}
\usepackage{subcaption}

\newtheorem{theorem}{Theorem}
\newtheorem{assumption}{Assumption}
\newtheorem{Remark}{Remark}

\newtheorem{Definition}{Definition}

\newtheorem{lemma}{Lemma}
\newtheorem{Proposition}{Proposition}
\usepackage{comment}
\usepackage{booktabs,color,amssymb,amsmath}

\newcommand{\circled}[1]{\mathbin{\text{\textcircled{\scriptsize #1}}}}

\def\BibTeX{{\rm B\kern-.05em{\sc i\kern-.025em b}\kern-.08em
    T\kern-.1667em\lower.7ex\hbox{E}\kern-.125emX}}
    
\begin{document}

\title{ {Differentially Private Secure Multiplication: Beyond Two Multiplicands}

\author{Haoyang Hu,
Viveck R. Cadambe}

\thanks{
Haoyang Hu and Viveck R. Cadambe are with the School of Electrical and Computer Engineering, Georgia Institute of Technology, Atlanta, GA, 30332 USA. (E-mail: \{haoyang.hu, viveck\}@gatech.edu)
This work is partially funded by the National Science Foundation under Grant CIF 2506573.
}

}

\maketitle

\begin{abstract}
We study the problem of differentially private (DP) secure multiplication in distributed computing systems, focusing on regimes where perfect privacy and perfect accuracy cannot be simultaneously achieved. 
Specifically, $\mathsf{N}$ nodes collaboratively compute the product of $\mathsf{M}$ private inputs while guaranteeing $\epsilon$-DP against any collusion of up to $\mathsf{T}$ nodes. Prior work has characterized the fundamental privacy–accuracy trade-off for the multiplication of two multiplicands. In this paper, we extend these results to the more general setting of computing the product of an arbitrary number $\mathsf{M}$ of multiplicands.
We propose a secure multiplication framework based on carefully designed encoding polynomials combined with layered noise injection. The proposed construction generalizes existing schemes and enables the systematic cancellation of lower-order noise terms, leading to improved estimation accuracy. We explore two regimes: $(\mathsf{M}-1)\mathsf{T}+1 \le \mathsf{N} \le \mathsf{M}\mathsf{T}$ and $\mathsf{N}=\mathsf{T}+1$. 
For $(\mathsf{M}-1)\mathsf{T}+1 \le \mathsf{N} \le \mathsf{M}\mathsf{T}$, we characterize the optimal privacy--accuracy trade-off. 
When $\mathsf{N}=\mathsf{T}+1$, we derive nontrivial achievability and converse bounds that are asymptotically tight in the high-privacy regime.
\end{abstract}

\section{Introduction}

Secure multi-party computation (MPC) enables multiple parties to collaboratively compute functions over their private inputs while preserving privacy \cite{goldreich1998secure}. In this paper, we focus on secure computation of products involving $\mathsf{M}$ multiplicands, where parties hold private random variables $A_1, A_2, \ldots, A_\mathsf{M}$ and aim to compute their product $\prod_{i=1}^\mathsf{M} A_i$ without revealing the individual inputs. This computation arises naturally in numerous applications, including secure computation of high-dimensional statistics, multivariate moments, and complex polynomial functions \cite{du2001secure}. Existing information-theoretically secure MPC protocols (e.g., the celebrated BGW protocol \cite{goldreich1998secure} and variations) enable secure computation for $\mathsf{N}$ computation nodes with up to $\mathsf{T}$ potentially colluding nodes under two scenarios: (a) an honest majority setting ($\mathsf{N} \geq 2\mathsf{T}+1$, where most nodes are honest) using an interactive protocol with $O(\mathsf{M})$ rounds, or (b) a one-round protocol requiring $\mathsf{N} \geq \mathsf{M}\mathsf{T}+1$ nodes. These approaches ensure perfect privacy and perfect accuracy—that is, the information gathered by any subset of $\mathsf{T}$ nodes during the protocol is statistically independent of the $\mathsf{M}$ random variables $A_1, A_2,\ldots, A_\mathsf{M}$, and the product $\prod_{i=1}^\mathsf{M} A_i$ can be exactly recovered by a legitimate user who can receive the computation output from all $\mathsf{N}$ nodes. However, they incur infrastructural resource or communication overhead that scales linearly with $\mathsf{M}$. This creates a fundamental bottleneck for secure MPC in modern machine learning tasks involving complex nonlinear computations over high-dimensional data. We address this challenge by studying an information-theoretic framework that relaxes the requirements of perfect privacy and accuracy. The framework uses differential privacy (DP) to enable controlled privacy leakage and characterizes privacy-accuracy trade-offs for one-round protocols in the resource-constrained regime of fewer than $\mathsf{M}\mathsf{T}+1$ nodes.

There is an extensive literature on the privacy–accuracy trade-off in DP that systematically calibrates noise distributions to satisfy DP constraints while meeting specific utility objectives\cite{kairouz2014extremal, geng2015optimal, geng2018privacy, balle2018improving, gilani2025optimizing}. In particular, a significant body of prior work focuses on marginal mechanisms, such as additive noise mechanisms optimized for $\ell_1$ or $\ell_2$ loss, and selects per query or per node noise distributions to minimize the expected error under a fixed privacy budget. 
In distributed computations involving multiple nodes and potential collusion, controlling only marginal noise distribution is insufficient, and carefully designed noise correlations can significantly improve accuracy \cite{imtiaz2019distributed, vithana2025correlated, pillutla2025correlated, choquettecorrelated, koloskova2023gradient, allouah2024privacy}.
In classical secure MPC over finite fields, these correlations are usually imposed via algebraic constructions such as Reed–Solomon codes and Shamir secret sharing, which guarantee that the sensitive data is statistically independent from the input to designated subsets of users.
For secure multiplication over the reals in the honest-minority regime $\mathsf{N} \leq 2\mathsf{T}$, where DP (rather than statistical independence) is the target, recent results \cite{cadambe2023differentially, hu2025differentially} establish the optimal \emph{joint} noise distribution across the nodes. This noise distribution is inherently correlated, and the correlation structure arises from a careful weaving of previously studied DP-optimal mechanisms with real-valued embeddings of generalized Reed–Solomon (GRS) codes from coding theory.
Specifically, for unit-variance inputs, this construction achieves the following asymptotic mean squared error in the multiplication of two multiplicands $A_1$, $A_2$ (i.e., $\mathsf{M}=2$):
\begin{align}\label{eq:tradeoff} \mathbb{E}[(A_1A_2 - \tilde{V})^2] \to \frac{1}{(1+{\tt SNR}^*(\epsilon))^2} \end{align} where $\tilde{V}$ is the estimate of the product $A_1A_2$, and
${\tt SNR}^*(\epsilon)$ is a function of the DP parameter $\epsilon$.
The optimality of this tradeoff is established via a geometric converse, further enriching the growing literature on the power of correlated noise in DP mechanisms \cite{imtiaz2019distributed, vithana2025correlated, pillutla2025correlated, choquettecorrelated, koloskova2023gradient, allouah2024privacy}. 
However, prior work, including \cite{cadambe2023differentially, hu2025differentially}, is limited to $\mathsf{M}=2$ and thus addresses only a narrow subset of privacy-preserving MPC. The fundamental infrastructure bottlenecks emerge for complex functions when $\mathsf{M}$ is large. We take a critical next step by developing DP mechanisms for $\mathsf{M}>2$ that enable \emph{one-round} (contrary to conventional MPC) protocols with fewer than $\mathsf{M}\mathsf{T}+1$ nodes, even for complex computations.


In this paper, we study two regimes: (1) $(\mathsf{M}-1)\mathsf{T} + 1 \leq \mathsf{N} \leq \mathsf{M}\mathsf{T}$ and (2) $\mathsf{N}=\mathsf{T}+1$. In the first regime of $(\mathsf{M}-1)\mathsf{T} + 1 \leq \mathsf{N} \leq \mathsf{M}\mathsf{T}$, for unit variance $A_1,\dots,A_\mathsf{M}$, we show that the privacy-accuracy trade-off satisfies the converse:
\begin{align}
\mathbb{E}\left[\left(\prod_{i=1}^{\mathsf{M}}A_i - \tilde{V}\right)^2\right] \geq \frac{1}{(1+{\tt SNR}^*(\epsilon))^\mathsf{M}}
\end{align}
where $\tilde{V}$ is an estimate obtained by the decoder, and we show that this trade-off is asymptotically achievable. The optimal trade-off is derived by first developing an alternate, geometric interpretation of \eqref{eq:tradeoff}. The geometric interpretation leads to a natural generalization of the trade-off of \eqref{eq:tradeoff} for the $(\mathsf{M}-1)\mathsf{T} + 1 \leq \mathsf{N} \leq \mathsf{M}\mathsf{T}$ regime. The second regime of $\mathsf{N} = \mathsf{T}+1$ is important as it represents minimal redundant infrastructure. Achieving privacy-preserving computation of non-linear functions with just one round and good accuracy can have a high impact in practice. 
For this practically important regime, we study the case of $\mathsf{N} < \mathsf{M}$ and establish constructive privacy mechanism design results together with fundamental impossibility bounds for the privacy–accuracy trade-off. While a gap remains between the bounds in this regime, we show that our bounds are information-theoretically tight in the high-privacy limit, that is, as $\epsilon \to 0$.

\emph{Related Work:} 
Correlated-noise mechanisms have been explored to improve utility in several DP settings \cite{imtiaz2019distributed, vithana2025correlated, pillutla2025correlated, choquettecorrelated, koloskova2023gradient, allouah2024privacy}. For distributed mean estimation and general distributed function computation, early work and recent advances design cross-user correlated Gaussian noise to reduce variance while preserving local or distributed DP guarantees \cite{imtiaz2019distributed, vithana2025correlated}. In iterative optimization and empirical risk minimization, a line of work \cite{pillutla2025correlated, choquettecorrelated, koloskova2023gradient} shows that linearly correlated noise can provably outperform independent noise in private learning, with both theoretical convergence analyses and scalable mechanisms. In decentralized and federated learning, correlated and often canceling noise is combined with secure aggregation to approach centralized DP utility without relying on a trusted curator \cite{allouah2024privacy}.

Coding techniques are widely used in distributed systems to protect sensitive inputs \cite{yu2019lagrange, d2020gasp, jahani2023swiftagg+, akbari2021secure, chang2018capacity, jia2021capacity, liang2024privacy, soleymani2021coded}. However, these works operate over finite fields and require perfect privacy guarantees, which in turn necessitate additional infrastructure.
In addition to \cite{cadambe2023differentially, hu2025differentially}, a series of recent studies \cite{soleymani2021analog, soleymani2022analog, liu2023analog, liu2023differentially, makkonen2022analog, makkonen2025analog, borah2024securing} have investigated private distributed computing directly over the real domain, where a perfect information-theoretic privacy guarantee is generally unattainable, and like us, approximation is necessary.
In particular, \cite{liu2023analog, liu2023differentially} adopt DP to characterize the privacy–utility trade-off in multiparty computation over real fields, similar in spirit to the present work. Their schemes rely on complex-valued Shamir’s secret sharing; however, \cite{cadambe2023differentially} demonstrates that these constructions fail to achieve the optimal privacy–accuracy trade-off.

To handle general nonlinear computations, the BGW protocol employs secret resharing \cite{wigderson1988completeness} by re-encoding intermediate results using freshly sampled polynomials of degree at most $\mathsf{T}$. While this preserves correctness and privacy and prevents growth in the number of nodes, it incurs additional communication rounds and computational overhead, significantly reducing efficiency.
Alternative approaches, such as Beaver multiplication triples (MTriples) \cite{beaver1992efficient}, enable more efficient secure multiplication using precomputed correlated randomness, but rely on an expensive offline preprocessing phase that may be impractical in many settings. Moreover, our goal is not merely to reduce the number of participating nodes required for privacy, but to develop a principled framework that extends naturally to more general computations. Designing efficient, secure multiplication protocols with fewer participating nodes, therefore, remains a fundamental open problem in secure MPC.


\emph{Notations}: 
Calligraphic symbols denote sets, bold symbols denote matrices and vectors, and sans-serif symbols denote system parameters.
$\vec{1}$ denotes an all-ones column vector.
For a positive integer $a$, we let $[a]\triangleq \{1,\dots,a\}$.
For a vector $\vec{\Gamma}$, we denote by $\vec{\Gamma}[i]$ the element with index $i$,.
For a matrix $\mathbf{A}$, let $\mathbf{A}^T$ denote its transpose. 
For functions $f$ and $g$ and for all large enough values of $x$, we write $f(x)=O(g(x))$ if there exists a positive real number $M$ and a real number $a_0 \in \mathbb{R}$ such that
$|f(x)|\le M |g(x)|$ for all $x\ge a_0$.

\begin{figure}[t]
    \centering
    \includegraphics[scale=0.5]{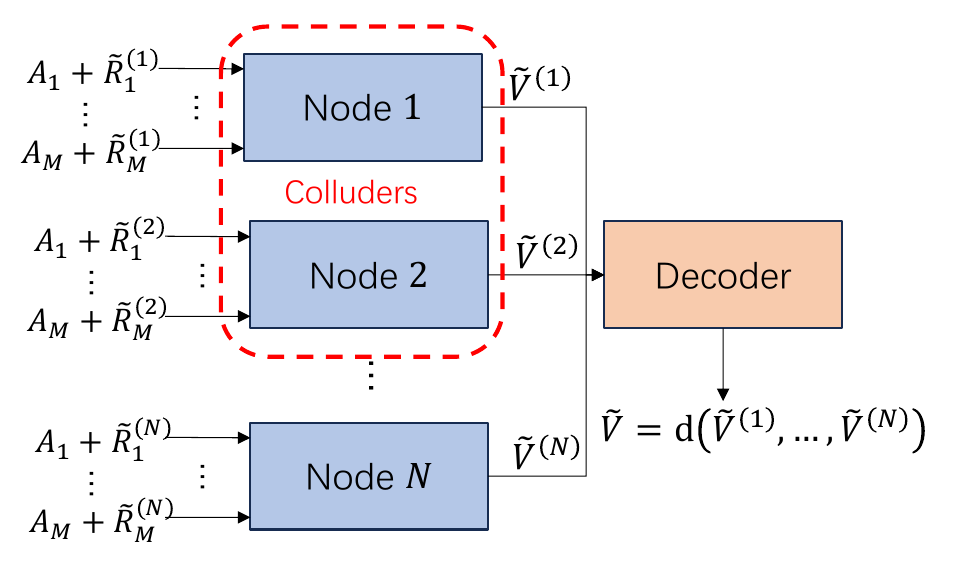}
    \caption{Illustration of the system model, where $\tilde{V}^{(j)} = \prod_{i \in [\mathsf{M}]} \bigl(A_i + \tilde{R}_i^{(j)}\bigr)$.}
    \label{fig::system_model}
\end{figure}

\section{System Model and Main Results}
In this section, we introduce the system model and present the main results of this paper. 

\subsection{System Model}
As shown in Figure \ref{fig::system_model}, we consider a distributed computing system with $\mathsf{N}$ nodes that collaboratively compute the product of $\mathsf{M}$ real-valued random variables $A_i \in \mathbb{R}$, where $i \in [\mathsf{M}]$. These random variables are assumed to satisfy the following condition\footnote{The privacy analysis does not rely on this assumption and the distribution of inputs $A_i$; the assumption is introduced for the purpose of evaluating the estimation accuracy.}.

\begin{assumption}
The random variables $A_i \in \mathbb{R}$ are statistically independent and satisfy
\begin{align}
    \mathrm{Var}(A_i) \le \eta,  
\end{align}
for a constant $\eta \ge 0$.
\label{assumption1}
\end{assumption}

Let each node $j\in[\mathsf{N}]$ store a noisy version of inputs $\{A_{i}\}_{i \in [\mathsf{M}]}$, given by
\begin{align}
    &\tilde{A}_i^{(j)} = A_i+\tilde{R}_i^{(j)}, \qquad i \in [\mathsf{M}], 
\end{align}
where $\{\tilde{R}_i^{(j)}\}_{i\in [\mathsf{M}], j \in [\mathsf{N}]}$ are  random variables that are statistically independent of the inputs $\{A_{i}\}_{i \in [\mathsf{M}]}$.
We assume that $\{\tilde{R}_1^{(j)}\}_{j \in [\mathsf{N}]}, \{\tilde{R}_2^{(j)}\}_{j \in [\mathsf{N}]},\dots, \{\tilde{R}_\mathsf{M}^{(j)}\}_{j \in [\mathsf{N}]}$ are statistically independent, i.e.,
\begin{align}
    \mathbb{P}_{\{\tilde{R}_i^{(j)}\}_{i\in [\mathsf{M}], j \in [\mathsf{N}]}} = \prod_{i\in [\mathsf{M}]} \mathbb{P}_{\{\tilde{R}_i^{(j)}\}_{j \in [\mathsf{N}]}}.
\end{align}

Without loss of generality, we assume that inputs $\{A_i\}_{i\in [\mathsf{M}]}$ and random variables $\{\tilde{R}_i^{(j)}\}_{i\in [\mathsf{M}], j \in [\mathsf{N}]}$ have zero mean \footnote{If the mean is nonzero, one can simply replace the decoding function with an affine mapping. It is straightforward to verify that all of our results continue to hold without the zero-mean assumption when affine decoders are used.}.

Node $j\in[\mathsf{N}]$  outputs the product of its local data, i.e.,
\begin{align}
    \tilde{V}^{(j)}= \prod_{i\in [\mathsf{M}]} \tilde{A}_i^{(j)} = \prod_{i\in [\mathsf{M}]} \left( A_i+\tilde{R}_i^{(j)}\right).
\end{align}

A decoder receives the computation output from all nodes, and then applies a linear decoding function $d: \mathbb{R}^\mathsf{N}\rightarrow \mathbb{R}$ to estimate the desired product $\prod_{i\in [\mathsf{M}]} A_i$, i.e., the decoder outputs 
\begin{align}
    \tilde{V}=d( \tilde{V}^{(1)}, \dots, \tilde{V}^{(\mathsf{N})}) = \sum_{j=1}^{\mathsf{N}} d_j \tilde{V}^{(j)},
\end{align}
where the coefficients $d_j \in \mathbb{R}$ specify the linear function $d$.

We assume that the decoding function can be designed based on the knowledge of the joint distributions of $\{\tilde{R}_i^{(j)}\}_{i\in [\mathsf{M}], j \in [\mathsf{N}]}$ and the parameters $\mathsf{N}, \mathsf{T}, \mathsf{M}, \eta$.
A secure multiplication coding scheme $\mathcal{C}(\mathsf{N}, \mathsf{T}, \mathsf{M}, \eta)$ specifies both the joint distribution of $\{\tilde{R}_i^{(j)}\}_{i\in [\mathsf{M}], j \in [\mathsf{N}]}$ and a linear decoding function $d: \mathbb{R}^\mathsf{N}\rightarrow \mathbb{R}$. We omit the dependence on $\mathcal{C}$ in $(\mathsf{N}, \mathsf{T}, \mathsf{M}, \eta)$ when clear from context.
A coding scheme $\mathcal{C}(\mathsf{N}, \mathsf{T}, \mathsf{M}, \eta)$ satisfies $\mathsf{T}$-node $\epsilon$-DP if the data at any $\mathsf{T}$ nodes is $\epsilon$-DP with respect to the original inputs.
Mathematically:

\begin{Definition}[$\mathsf{T}$-node $\epsilon$-Differential Privacy ($\mathsf{T}$-node $\epsilon$-DP)]\label{def::dp}
    A coding scheme $\mathcal{C}(\mathsf{N}, \mathsf{T}, \mathsf{M}, \eta)$ with  random noise variables $\{\tilde{R}_i^{(j)}\}_{i\in [\mathsf{M}], j \in [\mathsf{N}]}$ satisfies $\mathsf{T}$-node $\epsilon$-DP for $\epsilon>0$ if for each $i \in [\mathsf{M}]$, any $A_{i,0}, A_{i,1} \in \mathbb{R}$ satisfying $|A_{i,0} - A_{i,1}| \le 1$, we have  
    \begin{align}
        \frac{\mathbb{P} \left( \mathbf{X}^{\mathcal{T}}_{i,0} \in \mathcal{B}\right)}{\mathbb{P} \left( \mathbf{X}^{\mathcal{T}}_{i,1}\in \mathcal{B}\right)} \le e^{\epsilon},
    \end{align}
    for all subsets $\mathcal{T} \subseteq [\mathsf{N}]$ with $|\mathcal{T}|=\mathsf{T}$, and for all subsets $\mathcal{B}\subset \mathbb{R}^{1\times \mathsf{T}}$ in the Borel $\sigma$-field, where 
    \begin{align}
        & \mathbf{X}^{\mathcal{T}}_{i, \ell} \triangleq \begin{bmatrix}
            A_{i,\ell}+ \tilde{R}_i^{(t_1)} &  A_{i,\ell}+ \tilde{R}_i^{(t_2)} & \cdots & A_{i,\ell}+ \tilde{R}_i^{(t_\mathsf{T})}
        \end{bmatrix},
    \end{align}
    with $\ell \in \{0,1\}$, $\mathcal{T}=\{t_1, t_2, \dots, t_{\mathsf{T}}\}$.
\end{Definition}

We measure the accuracy of a coding scheme $\mathcal{C}$ in terms of its linear mean squared error.

\begin{Definition}[Linear Mean Square Error ($\tt LMSE$)]
Let $\mathcal{P}_{\eta}$ denote the set of all joint distributions 
$\mathbb{P}_{A_1\cdots A_\mathsf{M}}$ such that $A_1,\ldots,A_\mathsf{M}$ 
are mutually independent and $\mathrm{Var}(A_i) \le \eta$ for all 
$i\in[\mathsf{M}]$ (Assumption \ref{assumption1}). For a coding scheme $\mathcal{C}$ consisting of the 
joint distribution of $\{\tilde{R}_i^{(j)}\}_{i\in [\mathsf{M}], j \in [\mathsf{N}]}$ 
and a linear decoding function $d$, the $\tt LMSE$ is defined as
\begin{align}
    {\tt LMSE}(\mathcal{C}) 
    = \sup_{\mathbb{P}_{A_1\cdots A_\mathsf{M}} \in \mathcal{P}_{\eta}} 
    \mathbb{E}\left[\left|\tilde{V}-\prod_{i\in [\mathsf{M}]} A_i\right|^2\right],
\end{align}
where the expectation is taken jointly over $\{A_i\}_{i\in[\mathsf{M}]}$ 
distributed according to $\mathbb{P}_{A_1\cdots A_\mathsf{M}}$ and the 
randomness in $\{\tilde{R}_i^{(j)}\}_{i\in [\mathsf{M}], j \in [\mathsf{N}]}$.
\end{Definition}



For fixed parameters $\mathsf{N}, \mathsf{T}, \mathsf{M}, \eta$, we are interested in studying the trade-off between the ${\tt LMSE}$ and the $\mathsf{T}$-node DP parameter $\epsilon$ for coding schemes $\mathcal{C}(\mathsf{N}, \mathsf{T}, \mathsf{M}, \eta).$

\subsection{Main Results}

For any DP parameter $\epsilon>0$, we introduce the notation
\begin{align}
    {\tt SNR}^*(\epsilon) \triangleq \frac{\eta}{\sigma^{*}(\epsilon)^2},
\end{align}
where
\begin{equation} \sigma^{*}(\epsilon)^2 = \frac{2^{2/3}e^{-2\epsilon/3}(1+e^{-2\epsilon/3})+e^{-\epsilon}} {(1-e^{-\epsilon})^2}.\label{equ::optimalsigma}
\end{equation}
$\sigma^{*}(\epsilon)^2$ can be interpreted as the smallest variance of additive noise that ensures $\epsilon$-DP, corresponding to the staircase mechanism \cite{geng2015optimal}. A rigorous characterization is provided in Lemma~\ref{lemma::pramod_results}.

We first present the achievability result.

\begin{theorem}
Consider positive integers $\mathsf{N}, \mathsf{T}, \mathsf{M}$ with $(\mathsf{M}-1)\mathsf{T}+1 \leq \mathsf{N \le \mathsf{M} \mathsf{T}}$. 
For any $\epsilon,\xi > 0,$ there exists a coding scheme $\mathcal{C}$ that achieves $\mathsf{T}$-node $\epsilon$-DP secure multiplication with
\begin{equation}
   {\tt LMSE} (\mathcal{C}) \leq \frac{\eta^\mathsf{M}}{\left(1+ {\tt SNR}^*(\epsilon)\right)^{\mathsf{M}}}+\xi.
\end{equation}
\label{thm::achievability}    
\end{theorem}

Theorem \ref{thm::achievability}, which is a generalization of the main result of \cite{cadambe2023differentially}, is proved in Section \ref{sec::achievable_scheme}.

We then present a converse result establishing a lower bound on the ${\tt LMSE}$ for DP secure multiplication.

\begin{theorem}
Consider positive integers $\mathsf{N}, \mathsf{T}, \mathsf{M}$ with $\mathsf{M}  \leq \mathsf{N} \leq \mathsf{M}\mathsf{T}$. 
For any coding scheme ${\mathcal{C}}$ that achieves $\mathsf{T}$-node $\epsilon$-DP secure multiplication, there exists a distribution $\prod_{i \in [\mathsf{M}]} \mathbb{P}_{A_i}$ that satisfies Assumption \ref{assumption1} and
\begin{equation} 
    {\tt LMSE} (\mathcal{C}) \geq \frac{\eta^\mathsf{M}}{\left(1+ {\tt SNR}^*(\epsilon)\right)^\mathsf{M}}.
\end{equation}
\label{thm::converse}
\end{theorem}

Theorem \ref{thm::converse}, which is a generalization of the main result of \cite{cadambe2023differentially}, is proved in Section \ref{sec::converse}.

In Fig.~\ref{fig::baseline}, we compare our main result with two baseline schemes motivated by secret sharing and DP literature: complex-valued Shamir secret sharing and independent additive noise. The results show that the proposed scheme consistently outperforms both baselines.
The first baseline applies a complex-valued embedding of Shamir secret sharing \cite{soleymani2021analog,liu2023analog}, followed by a conventional DP analysis. The second baseline introduces independent additive noise at each node, where the noise distribution is chosen according to the optimal staircase mechanism~\cite{geng2015optimal}. Detailed descriptions of both baseline schemes can be found in Appendix~A of~\cite{cadambe2023differentially}.

\begin{figure}[t]
    \centering
    \includegraphics[scale=0.5]{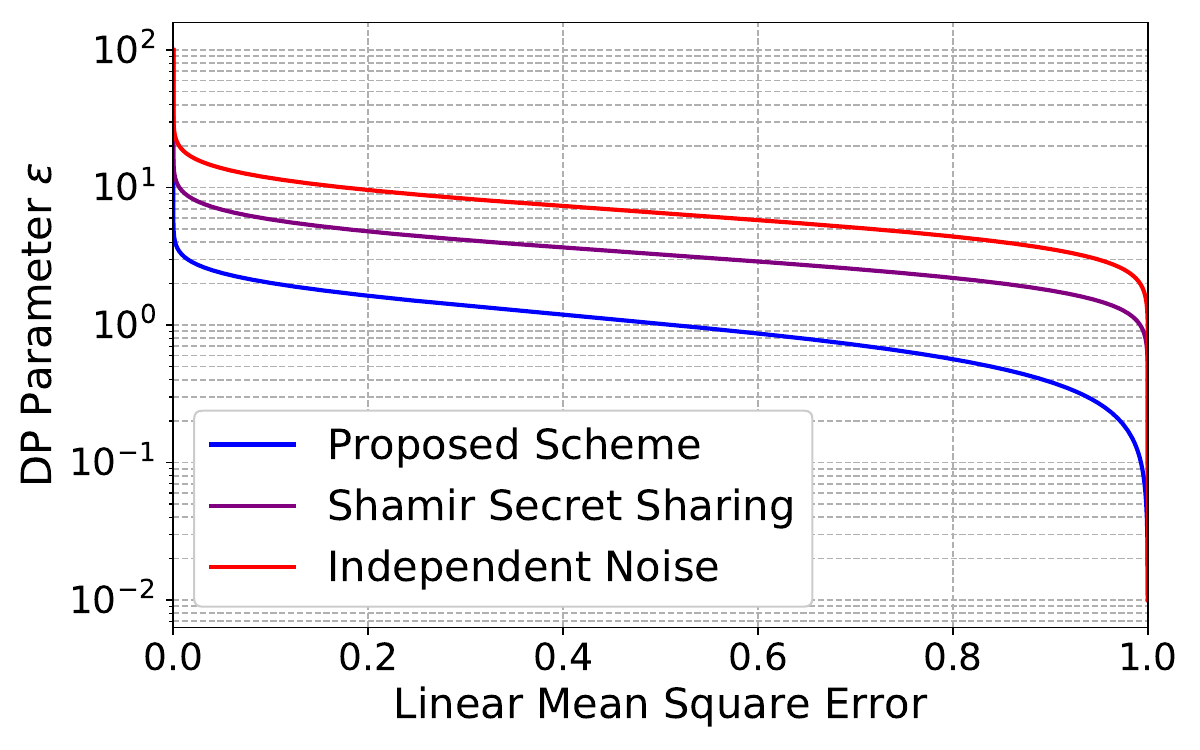}
    \caption{Performance of the proposed optimal scheme for $\mathsf{N}=5$, $\mathsf{M}=3$, and $\mathsf{T}=2$, compared with the baselines: (i) complex-valued Shamir secret sharing and (ii) independent noise across nodes.}
    \label{fig::baseline}
\end{figure}

From Theorems \ref{thm::achievability} and \ref{thm::converse}, we obtain a tight characterization of the privacy–accuracy trade-off in the regime $(\mathsf{M}-1)\mathsf{T}+1 \leq \mathsf{N} \leq \mathsf{M}\mathsf{T}$.
\cite{cadambe2023differentially} established, for the case of $\mathsf{M}=2$ multiplicands, the converse bound 
$
{\tt LMSE}(\mathcal{C}) \ge {\eta^2}/{\big(1+{\tt SNR}^*(\epsilon)\big)^2}
$
under the regime $\mathsf{N}\le 2\mathsf{T}$. 
Our converse result in Theorem~\ref{thm::converse} generalizes this bound to the case of $\mathsf{M} > 2$ multiplicands.
Furthermore, \cite{cadambe2023differentially} shows that when $\mathsf{T}+1 \le \mathsf{N} \le 2\mathsf{T}$, there exists an achievable scheme that attains the converse bound. This result can be regarded as a special case of our achievability result presented in Theorem~\ref{thm::achievability}.

Based on the achievability result in Theorem~\ref{thm::achievability}, secure multiplication of $\mathsf{M}$ multiplicands with $\mathsf{T}$-node $\epsilon$-DP can be achieved using $\mathsf{N} = (\mathsf{M}-1)\mathsf{T}+1$ nodes, which attains the converse bound in Theorem~\ref{thm::converse}. Consequently, in the regime $(\mathsf{M}-1)\mathsf{T}+1 \le \mathsf{N} \le \mathsf{M}\mathsf{T}$, increasing the number of participating nodes does not further reduce the achievable ${\tt LMSE}$.


We next consider the case $\mathsf{N} = \mathsf{T} + 1$ with $\mathsf{N} < \mathsf{M}$, for which neither the achievability condition in Theorem \ref{thm::achievability}, $(\mathsf{M}-1)\mathsf{T}+1 \le \mathsf{N} \le \mathsf{M}\mathsf{T}$, nor the converse condition in Theorem \ref{thm::converse}, $\mathsf{M} \le \mathsf{N} \le \mathsf{M}\mathsf{T}$, is satisfied.

\begin{theorem}
Consider the case $\mathsf{N}= \mathsf{T} + 1, \mathsf{N} < \mathsf{M}$. Then, for any $\epsilon,\xi > 0,$ there exists a coding scheme $\mathcal{C}$ that achieves $\mathsf{T}$-node $\epsilon$-DP secure multiplication with
\begin{equation}\label{equ::achievability_lessN}
   {\tt LMSE} (\mathcal{C}) \leq {\eta^\mathsf{M}} \frac{(1+{\tt SNR}^*(\epsilon))^{\mathsf{M}} - \mathsf{M} {\tt SNR}^*(\epsilon)^{\mathsf{M}-1} - {\tt SNR}^*(\epsilon)^{\mathsf{M}}}{\left(1 + {\tt SNR}^*(\epsilon)\right)^\mathsf{M}}+\xi.
\end{equation}
\label{thm::achievability_lessN}    
\end{theorem}

Theorem~\ref{thm::achievability_lessN} is proved in Appendix~\ref{sec::achievability_lessN}. We next establish a lower bound on the ${\tt LMSE}$.

\begin{theorem}

For the case $\mathsf{N}= \mathsf{T} + 1, \mathsf{N} < \mathsf{M}$. For any coding scheme ${\mathcal{C}}$ that achieves $\mathsf{T}$-node $\epsilon$-DP secure multiplication, there exists a distribution $\prod_{i \in [\mathsf{M}]} \mathbb{P}_{A_i}$ that satisfies Assumption \ref{assumption1} and
\begin{equation} 
    {\tt LMSE} (\mathcal{C}) \geq \eta^\mathsf{M} \frac{(1+{\tt SNR}^*(\epsilon))^{\mathsf{M}-\mathsf{T}} - {\tt SNR}^*(\epsilon)^{\mathsf{M}-\mathsf{T}}}{(1+{\tt SNR}^*(\epsilon))^\mathsf{M}}.
\end{equation}
\label{thm::converse_lessN}
\end{theorem}

Theorem~\ref{thm::converse_lessN} is proved in Appendix~\ref{sec::converse_lessN}. Note that a gap exists between the achievable upper bound and the converse lower bound for the case $\mathsf{N}= \mathsf{T} + 1$ with $\mathsf{N} < \mathsf{M}$. Closing this gap represents an interesting direction for future work. Specifically, the multiplicative gap is given by
\begin{align}
\mathsf{Gap}({\tt SNR}^*(\epsilon))
=
\frac{(1+{\tt SNR}^*(\epsilon))^{\mathsf{M}} - \mathsf{M} {\tt SNR}^*(\epsilon)^{\mathsf{M}-1} - {\tt SNR}^*(\epsilon)^{\mathsf{M}}}
{(1+{\tt SNR}^*(\epsilon))^{\mathsf{M}-\mathsf{T}} - {\tt SNR}^*(\epsilon)^{\mathsf{M}-\mathsf{T}}}.
\end{align}
In the high-privacy regime, where ${\tt SNR}^*(\epsilon) \to 0$ and subsequently $\epsilon \to 0$, a Taylor expansion yields
\begin{align}
\mathsf{Gap}({\tt SNR}^*(\epsilon))
= \frac{1+\mathsf{M}{\tt SNR}^*(\epsilon)+O({\tt SNR}^*(\epsilon)^2)}{1+(\mathsf{M}-\mathsf{T}){\tt SNR}^*(\epsilon)+O({\tt SNR}^*(\epsilon)^2)}
= 1 + \mathsf{T} {\tt SNR}^*(\epsilon) + O({\tt SNR}^*(\epsilon)^2).
\end{align}
This indicates that the bounds are tight in the high-privacy regime. 


\begin{figure}[t]
    \centering
    \begin{subfigure}{0.48\textwidth}
        \centering
        \includegraphics[scale=0.47]{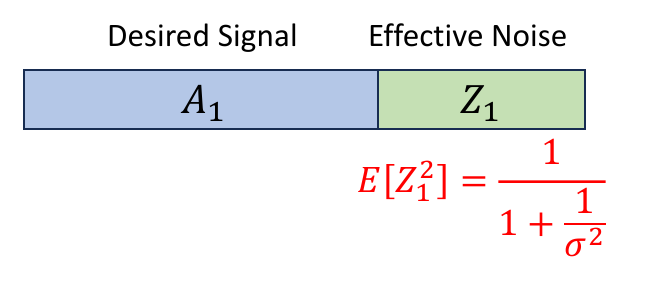}
        \caption{Illustration of estimating $A_1$.}
        \label{fig::example1}
    \end{subfigure}
    \begin{subfigure}{0.48\textwidth}
        \centering
        \includegraphics[scale=0.47]{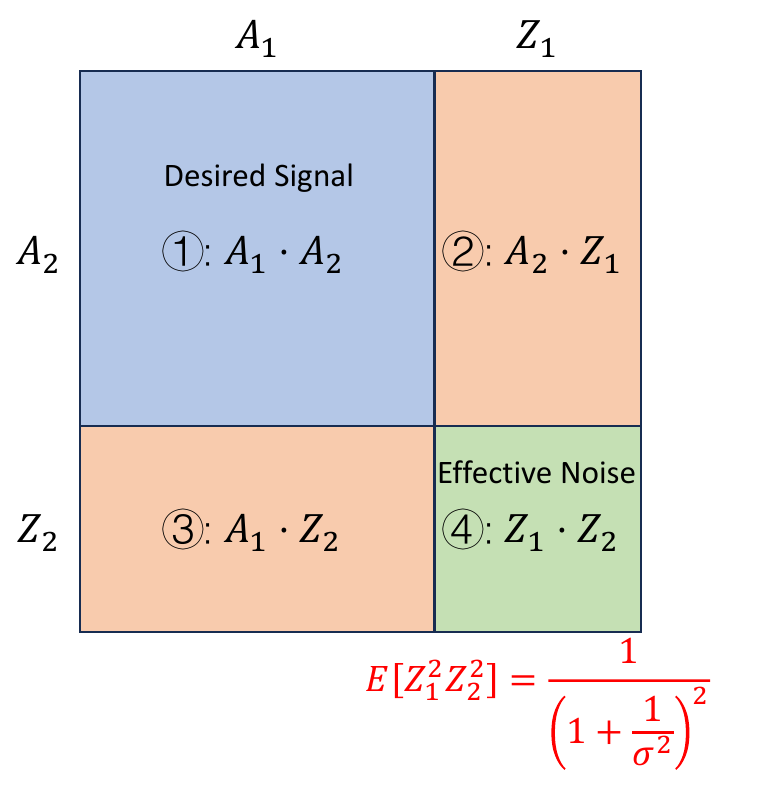}
        \caption{Illustration of estimating $A_1 \cdot A_2$ ($\circled{4}$ denotes the effective noise at the decoder)}
        \label{fig::example2}
    \end{subfigure}
    \caption{Geometric interpretation of Theorem \ref{thm::achievability} for $\mathsf{N}=2$ and $\mathsf{T}=1$, illustrating signal and noise components in the estimation of $A_1$ and $A_1 A_2$.}
\end{figure}

\section{Illustrative Examples}

This section illustrates the main ideas of this work through a series of representative examples. 
We start with the case $\mathsf{T}=1$, where each node aims to estimate the inputs based solely on its local data. 
First, we revisit the setting with two multiplicands ($\mathsf{M}=2$), which was studied in \cite{cadambe2023differentially}; we provide an intuitive geometric interpretation that facilitates extension to scenarios with more than two multiplicands. We then present examples with three multiplicands ($\mathsf{M}=3$), considering cases with either $\mathsf{T}=1$ or $\mathsf{T}=2$ colluding nodes. 
For simplicity, all inputs in this section are assumed to be normalized such that $\mathbb{E}[A_i^2] = 1$.

\subsection{$\mathsf{M}=2,\mathsf{N}=2, \mathsf{T}=1$}
\label{sec::first_example}
We consider a distributed setting with $\mathsf{N}=2$ nodes. Following the scheme in \cite{cadambe2023differentially},  node 1 receives
$\tilde{A}_1^{(1)}$ and $\tilde{A}_2^{(1)}$, and node 2 receives $\tilde{A}_1^{(2)}$ and $\tilde{A}_2^{(2)}$, where,   
for $\zeta>0$:
\begin{subequations}
\begin{align}
    &\tilde{A}_1^{(1)} = A_1 + R_1, 
    \quad 
    \tilde{A}_2^{(1)} = A_2 + R_2,\\
    & \tilde{A}_1^{(2)} = A_1 + (1 +\zeta) R_1, 
    \quad 
    \tilde{A}_2^{(2)} = A_2 + (1 +\zeta) R_2,
\end{align}    
\end{subequations}
$R_1, R_2$ are independent random variables, with the distribution of the $R_i$ chosen to satisfy two conditions:
(i) the privacy mechanism $A_i \mapsto A_i+ R_i$ ensures $\epsilon$-DP;
(ii) among all possible random variables satisfying (i), $R_i$ has the minimal possible variance, corresponding to the staircase mechanism \cite{geng2015optimal}.
The details for $R_i$ are specified in Lemma~\ref{lemma::pramod_results} and Section \ref{sec::dp_analysis}, and we denote its variance by $\mathbb{E}[R_i^2]=\sigma^2 \approx \sigma^{*}(\epsilon)^2$.
For $i=1, 2,$ the input $A_i$ satisfies $\epsilon$-DP due to the design of the random noise $R_i$.

The resulting computation outputs at the two nodes are, respectively, given by
\begin{subequations}\label{equ::first_example}
\begin{align}
    \tilde{V}^{(1)} &= \left(A_1 + R_1 \right) \left(A_2 + R_2\right), \\
    \tilde{V}^{(2)} &= \left(A_1 + (1 +\zeta) R_1\right) \left(A_2 + (1 +\zeta) R_2\right).
\end{align}    
\end{subequations}

\cite{cadambe2023differentially} shows that, under the scheme in \eqref{equ::first_example} and the $\epsilon$-DP constraint, the optimal linear mean squared error, denoted by ${\tt LMSE}^*$, is achieved asymptotically as  $\zeta \to 0$:
\begin{align}
    {\tt LMSE}^* =  \frac{1}{(1+\frac{1}{\sigma^2})^2} \approx \frac{1}{\left(1+ {\tt SNR}^*(\epsilon)\right)^2},
    \label{eq:LMSE_another}
\end{align}
where ${\tt SNR}^*(\epsilon)=\frac{1}{\sigma^{*}(\epsilon)^2}$. 
The original proof in \cite{cadambe2023differentially} relies on directly computing the signal-to-noise ratio of the target product, which involves calculating the determinant of the covariance matrix. However, this approach provides limited intuition and becomes cumbersome when generalizing to more multiplicands. We here present an alternative, more constructive proof sketch that (i) offers a clearer geometric and intuitive interpretation of the expression \eqref{eq:LMSE_another} and the corresponding decoding process, and (ii) naturally generalizes to multiplication involving a larger number of multiplicands.

First, consider a linear minimum mean-squared error (MMSE) estimator for the random variable $A_i$ from a noisy observation $A_i + R_i$:
\begin{align}
    \hat{A}_i = \alpha (A_i + R_i),
\end{align}
where $\alpha = \frac{1}{1+\sigma^ 2}$. 
Equivalently, we may write
\begin{align}
\hat{A}_i = A_i + Z_i,
\end{align}
with $Z_i=-\frac{\sigma^ 2}{1+\sigma^2} A_i + \frac{1}{1+\sigma^2} R_i$.
A direct calculation shows that
\begin{align}
    \mathbb{E}[(\hat{A}_i-A_i)^2] = \mathbb{E}[Z_i^2] = \frac{\sigma^4}{(1+\sigma^2)^2}+\frac{\sigma^2}{(1+\sigma^2)^2} = \frac{1}{1+\frac{1}{\sigma^2}} \approx \frac{1}{1+{\tt SNR}^*(\epsilon)}.
\end{align}
Thus, the estimation error for each individual $A_i$ is approximately $1/(1+{\tt SNR}^*(\epsilon))$, and $Z_i$ can be interpreted as the effective additive estimation noise after the linear MMSE scaling, as illustrated in Fig.~\ref{fig::example1}.

We now extend this single-variable perspective to analyze the estimation of the product $A_1 \cdot A_2$ from the two local products $\tilde V^{(1)}$ and $\tilde V^{(2)}$. From $\tilde V^{(1)}$, the decoder can construct
\begin{align}
    \alpha^2 \tilde{V}^{(1)} &= \alpha^2 (A_1 + R_1)(A_2 + R_2) = (A_1 + Z_1)(A_2 + Z_2).\\&= \underbrace{A_1 A_2}_{\circled{1}} + \underbrace{A_1 Z_2 + A_2 Z_1}_{\circled{2}+\circled{3}} + \underbrace{Z_1 Z_2}_{\circled{4}} 
\end{align}


The identity is visualized in Fig.~\ref{fig::example2} -- in effect, the quantity $\alpha^2\tilde{V}^{(1)} = (A_1+Z_1)(A_2+Z_2)$ can be interpreted as the area of a rectangle, which is the sum of the areas of four rectangles $A_1 A_2 \equiv \circled{1}, A_2 Z_1 \equiv \circled{2}, A_1 Z_2 \equiv \circled{3}$ and $Z_1 Z_2 \equiv \circled{4}.$   Observe that the rectangle $\circled{1}$ contains the desired quantitiy $A_1A_2$, and the optimal mean squared error $\frac{1}{(1+{1}/{\sigma^2})^2}$, is in fact, equal to $\mathbb{E}[(Z_1Z_2)^2]$. 
Therefore, to achieve this mean squared error up to a vanishing error, it suffices to show that the decoder can recover
$
\tilde{V} = A_1 A_2 - Z_1 Z_2 \equiv \circled{1} - \circled{4}
$
from $\tilde{V}^{(1)}$ and $\tilde{V}^{(2)}$ by canceling $\circled{2}$ and $\circled{3}$.




Using both $\tilde{V}^{(1)}$ and $\tilde{V}^{(2)}$, the decoder can compute,
$
\frac{\tilde{V}^{(2)}-\tilde{V}^{(1)}}{\zeta} = R_1(A_2+R_2)+R_2(A_1+R_1)+O(\zeta).
$
After scaling this by $\alpha$, the decoder can get
\begin{align}
& \quad \alpha \left(R_1(A_2+R_2)+R_2(A_1+R_1)\right) + O(\zeta) \\
&= R_1(A_2+Z_2)+R_2(A_1+Z_1)  + O(\zeta)\\
&=\frac{2}{\alpha} (A_1 + Z_1)(A_2 + Z_2) -A_1 (A_2+Z_2) -A_2 (A_1+Z_1) + O(\zeta).
\end{align}
Since $(A_1 + Z_1)(A_2 + Z_2)$ is already available from $\alpha^2\tilde V^{(1)}$, the decoder can recover the quantity $A_1 (A_2 + Z_2) + A_2 (A_1 + Z_1)$, which corresponds exactly to $2 \circled{1}+\circled{2} + \circled{3}$ in the rectangle notation.

Subtracting the term $2 \circled{1}+\circled{2} + \circled{3}$ from $\alpha^2 \tilde{V}^{(1)} \equiv \circled{1} + \circled{2} + \circled{3} + \circled{4}$ yields $A_1 A_2 - Z_1 Z_2 \equiv \circled{1} - \circled{4}$, up to an error of $O(\zeta)$, as desired. Consequently, a mean squared error arbitrarily close to $\mathbb{E}[Z_1^2 Z_2^2] = \frac{1}{(1+1/\sigma^2)^2}$ can be achieved.


\subsection{$\mathsf{M}=3,\mathsf{N}=3, \mathsf{T}=1$}\label{sec::example2}
We consider a previously unsolved computing setting where $\mathsf{N}=3$ nodes collaboratively compute the product of $\mathsf{M}=3$ private inputs $A_1, A_2, A_3$, subject to an $\epsilon$-DP constraint. Our encoding scheme generalizes the two-node construction described in \eqref{equ::first_example}, and for a fixed $\zeta > 0$, we set
\begin{subequations}
\begin{align}
    &\tilde{A}_1^{(1)} = A_1 + R_1, 
    \quad 
    \tilde{A}_2^{(1)} = A_2 + R_2, 
    \quad 
    \tilde{A}_3^{(1)} = A_3+ R_3,\\
    & \tilde{A}_1^{(2)} = A_1 + (1 +\zeta) R_1, 
    \quad 
    \tilde{A}_2^{(2)} = A_2 + (1 +\zeta) R_2, 
    \quad 
    \tilde{A}_3^{(2)} = A_3+ (1 +\zeta)R_3, \\
    & \tilde{A}_1^{(3)} = A_1 + (1 +2\zeta) R_1, 
    \quad 
    \tilde{A}_2^{(3)} = A_2 + (1 +2\zeta) R_2, 
    \quad 
    \tilde{A}_3^{(3)} = A_3+ (1 +2\zeta)R_3,
\end{align}    
\end{subequations}
where $\{R_i\}_{i=1}^3$ are independent random variables with $\mathbb{E}[R_i^2]=\sigma^2 \approx \sigma^{*}(\epsilon)^2$, chosen to minimize variance subject to the $\epsilon$-DP constraint. 
Similarly, the random variables $R_i$ serve to ensure single-node DP for each input.
Analogously to the two-node case, let $\alpha$ denote the single-variable MMSE scaling, and we have
\begin{align}
    A_i + Z_i = \alpha (A_i + R_i), \qquad i=1,2,3,
\end{align}
where $\mathbb{E}[Z_i^2]=1/(1+1/\sigma^2)$. 

Each node computes the product of its received inputs:
\begin{subequations}
\begin{align}
    \tilde{V}^{(1)} &= \left(A_1 + R_1 \right) \left(A_2 + R_2\right) \left(A_3+ R_3\right), \\
    \tilde{V}^{(2)} &= \left(A_1 + (1 +\zeta) R_1\right) \left(A_2 + (1 +\zeta) R_2\right) \left(A_3+ (1 +\zeta)R_3\right), \\
    \tilde{V}^{(3)} &= \left(A_1 + (1 +2\zeta) R_1\right) \left( A_2 + (1 +2\zeta) R_2 \right) \left( A_3+ (1 +2\zeta)R_3 \right).
\end{align}    
\end{subequations}

Based on $\tilde{V}^{(1)}$, $\tilde{V}^{(2)}$, and $\tilde{V}^{(3)}$, the decoder can derive
\begin{subequations}
\begin{align}
    \frac{\tilde{V}^{(2)}-\tilde{V}^{(1)}}{\zeta} 
    &= 
    (A_1+R_1)(A_2+R_2) R_3 
       + (A_1+R_1)(A_3+R_3) R_2 
       + (A_2+R_2)(A_3+R_3) R_1 
       + O(\zeta), \\
    \frac{\left(\tilde{V}^{(3)}-\tilde{V}^{(2)} \right) - \left(\tilde{V}^{(2)}-\tilde{V}^{(1)} \right)}{2\zeta^2} 
    &= (A_1+R_1) R_2 R_3 
       + (A_2+R_2) R_1 R_3 
       + (A_3+R_3) R_1 R_2 
       + O(\zeta).
\end{align}    
\end{subequations}

Taking the limit as $\zeta \to 0$, the decoder can therefore form the combinations
\begin{subequations}
\begin{align}
    C_0 &= (A_1+R_1)(A_2+R_2)(A_3+R_3)\\
    C_1 &= (A_1+R_1)(A_2+R_2) R_3 
       + (A_1+R_1)(A_3+R_3) R_2 
       + (A_2+R_2)(A_3+R_3) R_1, \\
    C_2 &= (A_1+R_1) R_2 R_3 
       + (A_2+R_2) R_1 R_3 
       + (A_3+R_3) R_1 R_2.
\end{align} \label{equ::Ck_example}
\end{subequations}
Applying the MMSE rescaling \(R_i = \frac{1}{\alpha}(A_i+Z_i) - A_i\) from Section~\ref{sec::first_example}, we substitute into \(C_0, C_1,\) and \(C_2\); expanding the resulting expressions then derives the corresponding linear relations.
\begin{subequations}
\begin{align}
C_0 &= \alpha^{-3} D_0,\\
C_1 &= 3\alpha^{-3} D_0 - \alpha^{-2} D_1,\\
C_2 &= 3\alpha^{-3} D_0 - 2\alpha^{-2}D_1 + \alpha^{-1}D_2,
\end{align}\label{equ::relation_Ck_Dk}
\end{subequations}
where 
\begin{subequations}
\begin{align}
D_0 &= (A_1+Z_1)(A_2+Z_2)(A_3+Z_3),  \\
D_1 &= (A_1+Z_1)(A_2+Z_2) A_3 
     + (A_1+Z_1)(A_3+Z_3) A_2 
     + (A_2+Z_2)(A_3+Z_3) A_1, \\
D_2 &= (A_1+Z_1) A_2 A_3 
     + (A_2+Z_2) A_1 A_3 
     + (A_3+Z_3) A_1 A_2.
\end{align}
\end{subequations}
\eqref{equ::relation_Ck_Dk} can be inverted since the coefficient matrix is triangular with nonzero diagonal entries. Consequently, the decoder can recover \(D_0, D_1,\) and \(D_2\) as
\begin{align}
D_0=\alpha^3 C_0,\qquad
D_1=\alpha^2(3C_0-C_1),\qquad
D_2=\alpha\big(C_2+3C_0-2C_1\big).
\end{align}


The decoder can then derive,
\begin{align}
    \tilde{V} = D_0 - D_1 + D_2 =  A_1 A_2 A_3 + Z_1 Z_2 Z_3,
\end{align}
Observe that this way, the effective noise is $\tilde{V}-A_1A_2A_3 = Z_1 Z_2 Z_3$.

Due to the independence between the residuals $Z_1, Z_2, Z_3$, the ${\tt LMSE}$ of the decoder is as follows.
\begin{align}
    \mathbb{E}[|\tilde{V} - A_1 A_2 A_3|^2] = \mathbb{E}[Z_1^2 Z_2^2 Z_3^2] = \mathbb{E}[Z_1^2] \mathbb{E}[Z_2^2] \mathbb{E}[Z_3^2] = \frac{1}{(1+\frac{1}{\sigma^2})^3} \approx \frac{1}{\left(1+ {\tt SNR}^*(\epsilon)\right)^3}.
\end{align}
Thus, the achievable ${\tt LMSE}$ equals $\frac{1}{(1+{1}/{\sigma^2})^3}$, which naturally generalizes the two-multiplicand result: the estimation error scales multiplicatively with the single-variable MMSE terms, since the final residual is the product of the individual estimation noises.

\subsection{$\mathsf{M}=3,\mathsf{N}=5, \mathsf{T}=2$}
We now extend the preceding analysis to the general case of $\mathsf{T}>1$, focusing on $\mathsf{T}=2$, and show that the achievable ${\tt LMSE}$ remains $\frac{1}{(1+{1}/{\sigma^2})^3}$, consistent with the result in Theorem \ref{thm::achievability}.
In this setting, any pair of $\mathsf{T}=2$ nodes may collude to infer a private input, and thus a $2$-node $\epsilon$-DP guarantee is required.  Under a standard Shamir secret sharing approach, achieving perfect information-theoretic privacy would require at least $\mathsf{M}\mathsf{T}+1 = 7$ nodes (for $\mathsf{M}=3$ and $\mathsf{T}=2$). In contrast, the proposed scheme ensures the 2-node $\epsilon$-DP guarantee with only $\mathsf{N} = 5$ nodes.

To construct the scheme, for $\zeta>0$, we define three encoding polynomials
\begin{subequations}
\begin{align}
f_1(x) &= A_1+ \underbrace{R_1}_{\text{First Layer}} + \underbrace{\zeta^{3/4} S_1 x}_{\text{Second Layer}} + \underbrace{\zeta R_1 x^2}_{\text{Third Layer}},\\
f_2(x) &= A_2+R_2+ \zeta^{3/4} S_2 x + \zeta R_2 x^2,\\ f_3(x) &= A_3+R_3+ \zeta^{3/4} S_3 x + \zeta R_3 x^2,
\end{align}
\end{subequations}
where $\{R_i,S_i\}_{i=1}^3, $ are mutually independent random variables, with the marginals of $R_{i}|_{i=1}^{3}$ chosen as the DP-calibrated noise variables with minimal variance as introduced in Section \ref{sec::first_example}, and $\{S_i\}_{i=1}^3$ random variables with unit variance. 
Each node $i \in \{1, \dots, 5\}$ stores the triple $(f_1(x_i), f_2(x_i), f_3(x_i))$, where $x_i = i$.

Note that the additive noise at each node can be interpreted as a superposition of three layers. Based on the distributions of $\{R_i\}_{i=1}^3$, the first and third layers implement an optimal staircase mechanism, while the second layer functions similarly to a secret sharing scheme. 
As shown in Section~\ref{sec::achievable_scheme}, in the limit $\zeta \to 0$, the proposed scheme asymptotically achieves the optimal privacy–accuracy trade-off.

We now provide an informal privacy analysis of the scheme.
Without loss of generality, we assume that nodes 1 and 2 collaborate to infer $A_1$, and our goal is to characterize the privacy loss. 
Under this setting, the available observations are given as follows.
\begin{subequations}
\begin{align}
\tilde{A}_1^{(1)} & = A_1+ {R_1} + {\zeta^{3/4} S_1 } + {\zeta R_1 }, \\
\tilde{A}_1^{(2)} & = A_1+ {R_1} + 2{\zeta^{3/4} S_1} + 4{\zeta R_1}.
\end{align}    
\end{subequations}
As $\zeta \to 0$, the magnitude of the third layer diminishes relative to the first two layers. Consequently, $\tilde{A}_1^{(1)}$ and $\tilde{A}_1^{(2)}$ become statistically close to 
$A_1 + R_1 + \zeta^{3/4} S_1$ and $A_1 + R_1 + 2 \zeta^{3/4} S_1$, respectively. 
In other words, $\tilde{A}_1^{(1)}$ and $\tilde{A}_1^{(2)}$ can be viewed as degraded versions of $A_1 + R_1$. 
Since $R_1$ is designed to satisfy $\epsilon$-DP, the scheme asymptotically achieves $2$-node $\epsilon$-DP.

We next outline the accuracy analysis. 
Define the product polynomial
\begin{align}\label{equ::product_polynomial}
    f(x) = f_1(x) f_2(x) f_3(x) = \sum_{k=0}^6 c_k x^k,
\end{align}
whose coefficients admit the expansion
{
\small
\begin{subequations}
\begin{align}
c_0 &= (A_1+R_1)(A_2+R_2)(A_3+R_3), \\[0.6ex]
c_1 &= \zeta^{3/4} \Big(
       (A_1+R_1)(A_2+R_2) S_3
     + (A_1+R_1)(A_3+R_3) S_2
     + (A_2+R_2)(A_3+R_3) S_1 \Big), \\[0.6ex]
c_2 &= \zeta \Big(
       (A_1+R_1)(A_2+R_2) R_3
     + (A_1+R_1)(A_3+R_3) R_2
     + (A_2+R_2)(A_3+R_3) R_1 \Big) \notag\\
     &\quad + \zeta^{3/2} \Big(
       (A_1+R_1) S_2 S_3
     + (A_2+R_2) S_1 S_3
     + (A_3+R_3) S_1 S_2 \Big), \\[0.6ex]
c_3 &= \zeta^{7/4} \Big(
       (A_1+R_1) S_2 R_3
     + (A_1+R_1) S_3 R_2
     + (A_2+R_2) S_1 R_3
     + (A_2+R_2) S_3 R_1
     + (A_3+R_3) S_1 R_2
     + (A_3+R_3) S_2 R_1 \Big) \notag\\
     &\quad + \zeta^{9/4} S_1 S_2 S_3, \\[0.6ex]
c_4 &= \zeta^{2} \Big(
       (A_1+R_1) R_2 R_3
     + (A_2+R_2) R_1 R_3
     + (A_3+R_3) R_1 R_2 \Big) + \zeta^{5/2} \Big(
       S_1 S_2 R_3
     + S_1 S_3 R_2
     + S_2 S_3 R_1 \Big), \\[0.6ex]
c_5 &= \zeta^{11/4} \Big(
       S_1 R_2 R_3
     + S_2 R_1 R_3
     + S_3 R_1 R_2 \Big), \\[0.6ex]
c_6 &= \zeta^{3} R_1 R_2 R_3.
\end{align}
\end{subequations}}


Recovering all coefficients $\{c_k\}_{k=0}^6$ would require $7$ evaluations of $f(x)$.  However, to estimate the target product $A_1 \cdot A_2 \cdot A_3$ with the optimal ${\tt LMSE}$, it suffices to recover only the coefficient triple $\{c_0,c_2,c_4\}$, as these coefficients contain the terms $C_0, C_1, C_2$ (defined in \eqref{equ::Ck_example}), the combinations used in the decoding procedure described in Section \ref{sec::example2}.  
As $\zeta \to 0$, the higher-order coefficients $c_5$ and $c_6$ become negligible relative to $c_0, \dots, c_4$, so that $\mathsf{N} = 5$ evaluations of $f(x)$ suffice to recover $c_0$, $c_2$, and $c_4$. Computing $c_2 / \zeta$ and $c_4 / \zeta^2$ then yields the desired symmetric sums in the $\zeta \to 0$ limit. Proceeding similarly to the $\mathsf{T} = 1$ case, one can show that the achievable ${\tt LMSE}$ is the optimal value $\frac{1}{(1+{1}/{\sigma^2})^3}$.

The $\mathsf{N}=5$ construction is efficient, relying on two non-trivial design choices. First, the careful scaling of powers of $\zeta$ ($\lim_{\zeta \to 0} (\zeta^{3/4})^2/\zeta =0$) separates the desired terms from higher-order contamination $(c_5, c_6)$, enabling the recovery of $(c_0,c_2,c_4)$ from only five evaluations. Second, the combination of DP-calibrated noise ($R_i$) and secret-sharing components ($\zeta^{3/4} S_i x$) simultaneously ensures privacy and decodability. Conceptually, the scheme mirrors a GRS code: nodes store polynomial evaluations, and decoding extracts the relevant low-order coefficients while higher-degree contributions vanish as $\zeta\to 0$, achieving the optimal ${\tt LMSE}$.


\section{Achievability: Proof of Theorem \ref{thm::achievability}}
\label{sec::achievable_scheme}

In this section, we present an $\mathsf{N}$-node secure multiplication coding scheme that achieves $\mathsf{T}$-node $\epsilon$-DP for the multiplication of $\mathsf{M}$ multiplicands, and the scheme is detailed in Section~\ref{sec::coding_schemes}.  
Specifically, we focus on the case $\mathsf{N} = (\mathsf{M}-1)\mathsf{T}+1 $.  
The $\mathsf{T}$-node $\epsilon$-DP guarantee is established in Section~\ref{sec::dp_analysis}.  
In Section~\ref{sec::accuracy}, we analyze the estimation accuracy by deriving the resulting ${\tt LMSE}$ for the product estimator. 

\subsection{Encoding Schemes}
\label{sec::coding_schemes}

For each $i\in [\mathsf{M}]$, let $R_i, \{ S_{i,t} \}_{t=1}^{\mathsf{T}-1}$ be mutually statistically independent random variables with zero mean. 
The specific distributions of $R_i, \{ S_{i,t} \}_{t=1}^{\mathsf{T}-1}$ will be specified in Section \ref{sec::dp_analysis} to ensure ${\mathsf{T}}$-node $\epsilon$-DP for the fixed DP parameter $\epsilon$.
We now introduce a sequence of coding schemes indexed by the positive integer $n$, that achieves the privacy-utility tradeoff described in Theorem \ref{thm::achievability} as $n\rightarrow \infty$.

Let $\zeta_1{(n)},\zeta_2{(n)}$ be strictly positive sequences such that:
\begin{equation} 
    \lim_{n \to \infty} \frac{\zeta_1{(n)}}{\zeta_2{(n)}} = \lim_{n \to \infty} \zeta_2{(n)} = \lim_{n\to\infty} \frac{\zeta_2{(n)}^{\mathsf{T}/(\mathsf{T}-1)}}{\zeta_1{(n)}} = 0.
    \label{eq:limits}
\end{equation} 
Observe that the conditions in \eqref{eq:limits} directly imply that $\lim_{n\to\infty} \zeta_1{(n)} = 0$.
A concrete example satisfying all conditions is $\zeta_2(n)=1/n$ and $\zeta_1(n)=1/n^{\beta_1}$ with $1 < \beta_1 < \mathsf{T}/(\mathsf{T}-1)$; e.g., $\beta_1 = (2\mathsf{T}-1)/(2(\mathsf{T}-1))$.

For the $i$-th multiplicand $A_i$, in the regime $\mathsf{T}>1$, we introduce the following polynomial:
\begin{align}\label{equ::encoding_polynomial}
    p_i(x) = (A_i + R_i) + \zeta_2(n)\sum_{t=1}^{\mathsf{T}-1} S_{i,t} x^t + \zeta_1(n) R_i x^\mathsf{T}.
\end{align}

In the special case $\mathsf{T}=1$, we introduce the following polynomial:
\begin{align}\label{equ::encoding_polynomial_t=1}
    p_i(x) = (A_i + R_i) + \zeta_1(n) R_i x.
\end{align}

Select $\mathsf{N}$ distinct real numbers $\{x_j\}_{j=1}^{\mathsf{N}}$. For each $j \in [\mathsf{N}]$, node $j$ receives a noisy version of the inputs given by
\begin{align}
    \tilde{A}_i^{(j)} = p_i(x_j), \qquad i  \in [\mathsf{M}].
\end{align}

The computation result of node $j \in [\mathsf{N}]$ is as follows.
\begin{align}
    \tilde{V}^{(j)}=\prod_{i=1}^{\mathsf{M}} p_i(x_j) = p(x_j),
\end{align}
where $p(x) = \prod_{i=1}^{\mathsf{M}} p_i(x)$ is the product polynomial, which is explicitly given in \eqref{equ::general_product_polynomial_T=1} and \eqref{equ::general_product_polynomial}.

\begin{Remark}
\label{remark::error_correction}
The coding scheme can be viewed as an $(\mathsf{N}, \mathsf{T}+1)$ real-valued GRS code with messages: 
$$\{A_i + R_i,\; S_{i,1},\; \ldots,\; S_{i,\mathsf{T}-1},\; R_i\},$$ 
and the corresponding GRS column multipliers (weights) are
$$
\{1,\; \zeta_2(n),\; \ldots,\; \zeta_2(n),\; \zeta_1(n)\}.
$$



\end{Remark}

\subsection{Differential Privacy Analysis}
\label{sec::dp_analysis}

Our proof is based on a specific realization of random variables $R_i, \{ S_{i,t} \}_{t=1}^{\mathsf{T}-1}$ for each $i \in [\mathsf{M}]$, and then establishing that the resulting scheme satisfies $\mathsf{T}$-node $\epsilon$-DP.
Due to the symmetry of the construction and the independence of the multiplicands, it suffices to only consider the privacy guarantee for the first multiplicand $A_1$. Specifically, we consider the worst-case scenario where an arbitrary set of $\mathsf{T}$ nodes colludes to infer the value of $A_1$.
A similar argument also applies to other multiplicands.

We begin by presenting a useful result from Theorem 7 in \cite{geng2015optimal}, which characterizes the minimal noise variance required to achieve single-node $\epsilon$-DP. This result is particularly relevant for the design of $R_i$.
\begin{lemma}[Theorem 7 in \cite{geng2015optimal}]\label{lemma::pramod_results}
For $\epsilon > 0,$ let $\mathcal{S}_{\epsilon}(\mathbb{P})$ denote the set of all real-valued random variables that satisfy $\epsilon$-DP, that is, $X \in\mathcal{S}_{\epsilon}(\mathbb{P})$ if and only if:
\begin{align}
    \sup \frac{\mathbb{P}(X+x' \in \mathcal{A})}{\mathbb{P}(X+x'' \in \mathcal{A})} \leq e^{\epsilon} 
\end{align}
where the supremum is over all constants $x',x'' \in \mathbb{R}$ that satisfy $|x'-x''| \leq 1$ and all subsets $\mathcal{A} \subset \mathbb{R}$ that are in the Borel $\sigma$-field. Let $L^2(\mathbb{P})$ denote the set of all real-valued random variables with finite variance. Then 
\begin{align}
    \inf_{X \in \mathcal{S}_\epsilon(\mathbb{P}) \cap L^2(\mathbb{P})} \mathbb{E}\left[(X-\mathbb{E}[X])^2\right]= \sigma^{*}(\epsilon)^2, 
\end{align}
where $\sigma^{*}(\epsilon)^2$ is given in \eqref{equ::optimalsigma}.
\end{lemma}
Informally, $\sigma^{*}(\epsilon)^2$ denotes the smallest noise variance that achieves the single-node DP parameter $\epsilon$. Based on \eqref{equ::encoding_polynomial}, the data stored at node $j$, i.e., $A_1^{(j)}$, can be rewritten as follows.
\begin{align}
    \tilde{A}_1^{(j)} &=\left(A_1+R_1\right)+\zeta_2(n)\begin{bmatrix}
        S_{1,1}  & \cdots & S_{1, \mathsf{T}-1}\
    \end{bmatrix} \vec{g}_j + \zeta_1(n) h_j R_1,
\end{align}
where $\vec{g}_j = 
\begin{bmatrix}
    x_j & x_j^2 & \cdots & x_j^{\mathsf{T}-1}
\end{bmatrix}^T$ and $h_j = x_j^{\mathsf{T}}$.
Let $
\mathbf{G}= 
\begin{bmatrix}
    \vec{g}_1 & \vec{g}_2 & \cdots & \vec{g}_\mathsf{N} 
\end{bmatrix}^T 
$ of size $\mathsf{N} \times (\mathsf{T}-1)$, $\vec{h}=
\begin{bmatrix}
    h_1 & h_2 & \cdots & h_\mathsf{N}
\end{bmatrix}^T$, and then define the Vandermonde matrix $\mathbf{M}=\begin{bmatrix}
    \vec{1} & \mathbf{G} & \vec{h}
\end{bmatrix}_{\mathsf{N} \times (\mathsf{T}+1)}$. 
We will use the Vandermonde matrix property that every $(\mathsf{T}-1) \times (\mathsf{T}-1)$, $\mathsf{T} \times \mathsf{T}$ and $(\mathsf{T}+1)\times(\mathsf{T}+1)$ submatrix of $\mathbf{M}$ is invertible.

We begin by specifying the distributions of the independent noise variables $R_1, \{S_{1,t}\}_{t=1}^{\mathsf{T}-1}$.

\begin{lemma}\label{lemma::existence_Z}
For every  $\sigma^2 > \sigma^{*}(\epsilon)^2$, there exists a random variable $N^*$ with $\mathbb{E}[(N^*)^2]\le \sigma^2$ such that $A_1+N^*$ achieves $\bar{\epsilon}$-DP for some $\bar{\epsilon}< \epsilon$ and $A_1 \in \mathbb{R}$. 

\end{lemma}

\begin{proof}

The proof is in Appendix \ref{sec::proof_existence_Z}.

\end{proof}

Let the additive noise $R_1$ follow the same distribution as $N^*$ described in Lemma \ref{lemma::existence_Z}, and it follows that $A_1+R_1$ guarantees $\bar{\epsilon}$-DP. 
The noise variables $S_{1,1}, S_{1,2}, \cdots, S_{1, \mathsf{T}-1}$ are independently drawn from a unit-variance Laplace random distribution, and are independent of $R_1$.

For the case $\mathsf{T} \ge 2$, we assume without loss of generality that the first $\mathsf{T}$ nodes form a colluding set. Due to the inherent symmetry of the coding scheme, the argument presented below applies identically to any other subset of $\mathsf{T}$ colluding nodes.
The colluding nodes receive:
\begin{align}\label{equ::z}
\vec{Z}=(A_1 +R_1)\vec{1}+ 
    \bar{\mathbf{G}}
    \begin{bmatrix}
        \zeta_1(n) R_1 & \zeta_2(n) S_{1,1} & \cdots & \zeta_2(n)  S_{1, \mathsf{T}-1}
    \end{bmatrix}^T,  
\end{align}
where $\vec{1} \in \mathbb{R}^{\mathsf{T}}$ denotes the all-ones column vector of length $\mathsf{T}$,
and $\bar{\mathbf{G}}=
\begin{bmatrix}
    h_1 & h_2 & \cdots & h_{\mathsf{T}} \\
    \vec{g}_1 & \vec{g}_2 & \cdots & \vec{g}_\mathsf{T}
 \end{bmatrix}^T$ of size ${\mathsf{T} \times \mathsf{T}}$. 
Let $\vec{g'}_j^T$ with $j \in [\mathsf{T}]$ denote the $j$-th row of the matrix $\bar{\mathbf{G}}^{-1}$, and then $\vec{g'}_j^T\vec{1}$ represents the $j$-th element of the column vector $\bar{\mathbf{G}}^{-1}\vec{1}$.

\begin{lemma} \label{lemma::_z_z'}
For $\vec{Z}$ in \eqref{equ::z}, there exists a full rank matrix $\mathbf{P}$ such that $\mathbf{P}^{-1} \vec{Z}=\vec{Z}'= 
\begin{bmatrix}
    Z'_1 & Z'_2 & \cdots & Z'_\mathsf{T}
\end{bmatrix}^T$, where 
\begin{align}\label{equ::z_factored}
    Z'_1 &= A_1 + \left(1 + \frac{\zeta_1(n)}{\vec{g'}_1^T\vec{1}}\right) R_1,\nonumber\\
    Z'_j &= A_1 + \frac{\zeta_2(n)\bigl(\vec{g'}_1^T\vec{1} + \zeta_1(n)\bigr)}{\zeta_1(n)\,\vec{g'}_j^T\vec{1}}\, S_{1,j-1}, \quad j=2,\ldots,\mathsf{T}.
\end{align}

\end{lemma}

\begin{proof}
    The proof is in Appendix \ref{sec::proof_z_z'}
\end{proof}

Since $\vec{Z}=\mathbf{P}\vec{Z}'$ is a deterministic function of $\vec{Z}'$, the post-processing property of DP~\cite{dwork2014algorithmic} implies that $\vec{Z}$ satisfies at least the same DP guarantee as $\vec{Z}'$. Therefore, it suffices to establish that $\vec{Z}'$ satisfies $\epsilon$-DP.

$Z'_1$ is a perturbation of $A_1+R_1$. We have that, for any $\bar{\bar{\epsilon}}$ with
$\bar{\epsilon}< \bar{\bar{\epsilon}} < \epsilon$, as $n \to \infty$,
\begin{align}
    &\sup_{\mathcal{B}\subseteq \mathcal{B}(\mathbb{R}), -1<\lambda<1} \frac{\mathbb{P}\left(A_1 + \left(1 + \frac{1}{ \vec{g'}_1^T\vec{1}} \zeta_1(n) \right) R_1 \in \mathcal{B}\right)}{\mathbb{P}\left(A_1 + \left(1 + \frac{1}{ \vec{g'}_1^T\vec{1}} \zeta_1(n) \right) R_1+\lambda \in \mathcal{B}\right)}  \\
    =&
    \sup_{\mathcal{B}\subseteq \mathcal{B}(\mathbb{R}), -\frac{1}{1+\frac{1}{\vec{g'}_1^T\vec{1}}\zeta_1(n)}<\lambda<\frac{1}{1+\frac{1}{\vec{g'}_1^T\vec{1}}\zeta_1(n)}} \frac{\mathbb{P}(A_1+R_1 \in \mathcal{B})}{\mathbb{P}(A_1+R_1+\lambda \in \mathcal{B})}  \\
    \overset{(a)}{\le} & e^{\bar{\bar{\epsilon}}},
\end{align}
where $(a)$ holds as $\lim_{n\rightarrow \infty} \frac{1}{1+\frac{1}{\vec{g'}_1^T\vec{1}}\zeta_1(n)}=1$. Hence $Z'_1$  achieves $\bar{\bar{\epsilon}}$-DP.

Each $Z'_j$ ($j\ge 2$) adds an independent Laplace noise with coefficient $\frac{\zeta_2(n)}{\zeta_1(n)}\cdot \frac{\vec{g'}_1^T\vec{1}+\zeta_1(n)}{\vec{g'}_j^T\vec{1}}\to\infty$ by~\eqref{eq:limits}, each $Z'_j$ achieves $\frac{\zeta_1(n)\vec{g'}_j^T\vec{1}}{\zeta_2(n)(\vec{g'}_1^T\vec{1}+\zeta_1(n))}\sqrt{2}$-DP. By the composition theorem~\cite{dwork2014algorithmic}, $\vec{Z}'$ achieves
\[\left(\bar{\bar{\epsilon}}+\sqrt{2}\sum_{j=2}^{\mathsf{T}}\frac{\zeta_1(n)\vec{g'}_j^T\vec{1}}{\zeta_2(n)(\vec{g'}_1^T\vec{1}+\zeta_1(n))}\right)\text{-DP.}\]
Since $\lim_{n\to\infty}\zeta_1(n)/\zeta_2(n)=0$ by~\eqref{eq:limits}, $\vec{Z}'$ achieves $\epsilon$-DP as $n \to \infty$. Consequently, $\vec{Z}=\mathbf{P}\vec{Z}'$ satisfies $\epsilon$-DP by post-processing~\cite{dwork2014algorithmic}.

For the case $\mathsf{T}=1$, we have that, for any $\gamma''>0$ with $e^{\bar{\epsilon}}+\gamma'' \le  e^{\epsilon}$, as $n \to \infty$,
\begin{align}
    &\sup_{\mathcal{B}\subseteq \mathcal{B}(\mathbb{R}), -1<\lambda<1} \frac{\mathbb{P}\left(A_1 + \left(1 + \zeta_1(n) h_1 \right) R_1 \in \mathcal{B}\right)}{\mathbb{P}\left(A_1 + \left(1 + \zeta_1(n) h_1 \right) R_1+\lambda \in \mathcal{B}\right)}  \\
    =&
    \sup_{\mathcal{B}\subseteq \mathcal{B}(\mathbb{R}), -\frac{1}{1 + \zeta_1(n) h_1}<\lambda<\frac{1}{1 + \zeta_1(n) h_1}} \frac{\mathbb{P}(A_1+R_1 \in \mathcal{B})}{\mathbb{P}(A_1+R_1+\lambda \in \mathcal{B})}  \\
    \overset{(a)}{\le} & e^{\bar{\epsilon}}+\gamma'' \le  e^{\epsilon},
\end{align}
where $(a)$ holds as $\lim_{n\rightarrow \infty} \frac{1}{1 + \zeta_1(n) h_1} =1$. 
Hence, the coding scheme achieves $\epsilon$-DP, completing the proof that the proposed scheme satisfies $\mathsf{T}$-node $\epsilon$-DP for any $\sigma^2 > \sigma^{*}(\epsilon)^2$.

\subsection{Accuracy Analysis}
\label{sec::accuracy}
This subsection aims to derive the ${\tt LMSE}$ of the encoding scheme presented in Section~\ref{sec::coding_schemes}. Based on Lemma~\ref{lemma::lmse} established in Section~\ref{sec::converse}, it suffices to consider the case $\mathbb{E}[A_i^2] = \eta$ in order to obtain an upper bound on ${\tt LMSE}$.

We examine the fundamental case of estimating a single random variable $A_i$ from the observation $A_i+R_i$. 
Based on the MMSE criterion, the corresponding optimal estimator is given by
\begin{align}
    \hat{A}_i = \alpha (A_i + R_i),
\end{align}
where $\alpha = \frac{\eta}{\sigma^2+\eta}$.
We also reformulate $\hat{A}_i$ as $\hat{A}_i = A_i + Z_i$,
where $Z_i = -\frac{\sigma^2}{\sigma^2+\eta} A_i + \frac{\eta}{\sigma^2+\eta} R_i$, and
\begin{align}
   \mathbb{E}[Z_i^2] = \frac{\eta \sigma^2}{\eta + \sigma^2} = \frac{\eta}{1+\frac{\eta}{\sigma^2}}. 
\end{align}
For integer $k=0,1, \cdots, \mathsf{M}-1$, we define
\begin{align}
    D_k = \sum_{\mathcal{S} \subseteq [\mathsf{M}], |\mathcal{S}|=k}\left( \prod_{i\in \mathcal{S}} A_i \right) \left( \prod_{l \notin \mathcal{S}} (A_l+Z_l) \right).
\end{align}
For example, $D_0=\prod_{i=1}^{\mathsf{M}} (A_i + Z_i)$ and
$D_1 = \sum_{i=1}^\mathsf{M} A_i \left( \prod_{l\in [\mathsf{M}], l\neq i} (A_l + Z_l) \right)$.

\begin{Proposition}\label{prop::CkDk}
Suppose $\mathsf{N} = (\mathsf{M}-1)\mathsf{T}+1$. For each $k=0,1,\ldots,\mathsf{M}-1$, there exists a function $f_k$ such that
\[
    \lim_{n\to\infty} 
    f_k\!\bigl(\tilde{V}^{(1)},\ldots,\tilde{V}^{(\mathsf{N})}\bigr)
    = D_k,
\]
where the coding scheme (and hence $\tilde{V}^{(j)}$) depends on $n$ through the parameters $\zeta_1(n),\zeta_2(n)$ defined in Section~\ref{sec::coding_schemes}.
In particular, $\{D_k\}_{k=0}^{\mathsf{M}-1}$ can be asymptotically recovered from the local computation results $\{\tilde{V}^{(j)}\}_{j\in[\mathsf{N}]}$, each with error vanishing as $n\to\infty$.
\end{Proposition}

\begin{proof}
The proof is in Appendix~\ref{sec::proof_CkDk}.    
\end{proof}

Based on Proposition~\ref{prop::CkDk}, the local computation results $\{\tilde{V}^{(j)}\}_{j\in [\mathsf{N}]}$ can recover $\{D_k\}_{k=0}^{\mathsf{M}-1}$ asymptotically.
For the sake of simplicity, the subsequent accuracy analysis is conducted based on $\{D_k\}_{k=0}^{\mathsf{M}-1}$ by proving the following proposition.

\begin{Proposition}\label{prop::alter_sum}
\begin{align}
    \sum_{k=0}^{\mathsf{M}-1} (-1)^k D_k \;=\; \prod_{i=1}^\mathsf{M} Z_i \;+\; (-1)^{\mathsf{M}+1} \prod_{i=1}^\mathsf{M} A_i.
\end{align}
\end{Proposition}

\begin{proof}

Expand $\prod_{i=1}^{\mathsf{M}} Z_i = \prod_{i=1}^{\mathsf{M}} \bigl((A_i+Z_i)-A_i\bigr)$:
\begin{align}
  \prod_{i=1}^{\mathsf{M}} Z_i = \sum_{\mathcal{S}\subseteq[\mathsf{M}]} (-1)^{|\mathcal{S}|} \prod_{i\in\mathcal{S}} A_i \prod_{l\notin\mathcal{S}} (A_l+Z_l) = \sum_{k=0}^{\mathsf{M}} (-1)^k D_k.
\end{align}
Isolating the $k=\mathsf{M}$ term ($(-1)^{\mathsf{M}} D_{\mathsf{M}} = (-1)^{\mathsf{M}}\prod_{i=1}^{\mathsf{M}} A_i$) and rearranging yields the result.

\end{proof}

If $\mathsf{M}$ is odd, then $(-1)^{\mathsf{M}+1} = 1$, so the alternating sum gives $S = \prod_{i=1}^\mathsf{M} A_i + \prod_{i=1}^\mathsf{M} Z_i$. If $\mathsf{M}$ is even, then $(-1)^{\mathsf{M}+1} = -1$, so $S = \prod_{i=1}^\mathsf{M} Z_i - \prod_{i=1}^\mathsf{M} A_i$.

Applying Proposition \ref{prop::alter_sum}, the mean estimation error can be bounded as
\begin{align}
    {\tt LMSE} \le \mathbb{E}\left[\left(\prod_{i=1}^\mathsf{M} Z_i\right)^2\right] \overset{(a)}{=} \prod_{i=1}^{\mathsf{M}} \mathbb{E}\left[Z_i^2\right] = \frac{\eta^\mathsf{M}}{\left(1+\frac{\eta}{\sigma^2}\right)^\mathsf{M}},
\end{align}
where $(a)$ follows from the independence of $Z_1, \dots, Z_\mathsf{M}$.
The achievable result in Theorem \ref{thm::achievability} is proved by substituting $\sigma^2 = \sigma^*(\epsilon)^2+\gamma'$ with sufficient small $\gamma'$ as shown in Section \ref{sec::dp_analysis}.

\section{Converse: Proof of Theorem \ref{thm::converse}}
\label{sec::converse}

For the converse proof, we first introduce compact tensor notation that will be used throughout the proof. For an order-\(\mathsf{M}\) tensor \(\mathcal{D}\in\mathbb{R}^{d_1\times\cdots\times d_{\mathsf{M}}}\) and a vector \(\vec{\alpha}\in\mathbb{R}^{d_k}\), the \emph{mode-\(k\) tensor--vector product} \(\mathcal{D}\times_k\vec{\alpha}\) denotes the order-\((\mathsf{M}-1)\) tensor obtained by contracting \(\mathcal{D}\) along its \(k\)-th index:
\(
(\mathcal{D}\times_k\vec{\alpha})_{i_1\ldots i_{k-1} i_{k+1}\ldots i_{\mathsf{M}}}
\;=\;\sum_{i_k=1}^{d_k}\mathcal{D}_{i_1\ldots i_k\ldots i_{\mathsf{M}}}\,\alpha_{i_k}.
\)

We present the following useful lemma, which is a well-known result from linear mean-square estimation theory~\cite{poor2013introduction}.
\begin{lemma}\label{lemma::lmse}
    Let $X$ be a random variable with $\mathbb{E}[X]=0$ and $\mathbb{E}[X^2]=\lambda^2$. Let $\{N_i\}_{i=1}^m$ be random noise variables independent of $X$, and
    $\tilde{\mathbf{X}}=\begin{bmatrix}
        \nu_1 X+ N_1 & \cdots & \nu_m X+ N_m
    \end{bmatrix}^T$, where $\nu_i \in \mathbb{R}$.
    There exists an explicit minimizer $\mathbf{w}^*$ that achieves equality in \eqref{equ::lmse_lemma}.
    \begin{align}\label{equ::lmse_lemma}
    \inf_{\mathbf{w}\in\mathbb{R}^{m}} \mathbb{E}[| \mathbf{w}^T \tilde{\mathbf{X}} -X|^2]=\frac{\lambda^2}{1+{\tt SNR}_a},    
    \end{align}
    where
    $
        {\tt SNR}_a=\frac{\det (\mathbf{K}_1)}{\det(\mathbf{K}_2)}-1
    $, where $\mathbf{K}_1$ denotes the covariance matrix of the noisy observation $\tilde{\mathbf{X}}$, and $\mathbf{K}_2$ denotes the covariance matrix of the noise $\{N_i\}_{i=1}^m$.
    Also, the vector $\mathbf{w}^*$ satisfies the following property: for any random variables $X'$ with  $\mathbb{E}[X']=0$ and $\mathbb{E}[X'^2] \le \lambda^2$, $\mathbb{E}[| \mathbf{w}^{*T} \tilde{\mathbf{X}}'-X'|^2] \le \frac{\lambda^2}{1+{\tt SNR}_a}$, where $\tilde{\mathbf{X}}'=\begin{bmatrix}
        \nu_1 X' + N_1 & \cdots & \nu_m X'+ N_m
    \end{bmatrix}^T$.
\end{lemma}

Consider the setting with $\mathsf{N}$ nodes with at most $\mathsf{T}$ colluding nodes.
In this section, we assume that $\mathbb{E}[A_i^2] = \eta$, and the result can be extended to the case $\mathbb{E}[A_i^2] < \eta$ by Lemma \ref{lemma::lmse}.
There exist uncorrelated, zero-mean, unit-variance random variables $\frac{A_i}{\sqrt{\eta} }, \bar{R}_i^{(1)}, \cdots, \bar{R}_i^{(\mathsf{N})}$ for each $i\in [\mathsf{M}]$\footnote{The set of random variables can be derived by multiplying the square root of the inverse of the covariance matrix of the additive noise across the nodes. 
}.
In this case, for each $i\in [\mathsf{M}]$, let
\begin{align}
    \vec{\Gamma}_i = 
    \begin{bmatrix}
    \frac{A_i}{\sqrt{\eta} } & \bar{R}_i^{(1)} & \cdots & \bar{R}_i^{(\mathsf{N})}    
    \end{bmatrix}^T\in\mathbb{R}^{\mathsf{N}+1}.
\end{align}


Node $j$ stores the noisy version of ${A_i}$ as follows.
\begin{align}
    \tilde{A}_i^{(j)} = \vec{\Gamma}_i^T \vec{w}_i^{(j)},
\end{align}
where $\vec{w}_i^{(j)} \in \mathbb{R}^{\mathsf{N}+1}$ is the corresponding linear combination coefficients.

In this case, node $j$ could get the following local computation results.
\begin{align}
    \tilde{V}^{(j)} = \prod_{i\in [\mathsf{M}]} \tilde{A}_i^{(j)} = \prod_{i\in [\mathsf{M}]} ( \vec{w}_i^{(j)})^T \vec{\Gamma}_i.
\end{align}

The decoder estimates the product by $\tilde{V} = \sum_{j=1}^{\mathsf{N}} d_j \tilde{V}^{(j)}$.
Define a rank-one tensor
\begin{align}
 \mathcal{A} = \vec{\Gamma}_1 \otimes \vec{\Gamma}_2 \otimes \cdots \otimes \vec{\Gamma}_{\mathsf{M}} \in \mathbb{R}^{\underbrace{\scriptstyle (\mathsf{N}+1) \times (\mathsf{N}+1)\times \cdots \times (\mathsf{N}+1)}_{\mathsf{M} \text{ times}}}, 
\end{align}
and let
\begin{align}
    \mathcal{D} = \sum_{j=1}^{\mathsf{N}} d_j
    \left({\vec{w}_1^{(j)} \otimes \vec{w}_2^{(j)} \otimes \cdots \otimes \vec{w}_{\mathsf{M}}^{(j)}} \right)- (\underbrace{\vec{e} \otimes \vec{e} \otimes \cdots \otimes\vec{e}}_{\mathsf{M} \text{ times}}) \in \mathbb{R}^{\underbrace{\scriptstyle (\mathsf{N}+1) \times (\mathsf{N}+1)\times \cdots \times (\mathsf{N}+1)}_{\mathsf{M} \text{ times}}},
\end{align}
where $\vec{e} = \begin{bmatrix}
        \sqrt{\eta} &
        0 &
        \cdots &
        0
    \end{bmatrix}^T \in \mathbb{R}^{\mathsf{N}+1}$.

Then, the estimation error can be expressed as a tensor inner product.
\begin{align}
    \tilde{V} - \prod_{i=1}^\mathsf{M} A_i &= \left(\sum_{j=1}^\mathsf{N} d_j \left( \prod_{i\in [\mathsf{M}]} ( \vec{w}_i^{(j)})^T \vec{\Gamma}_i \right) \right) -  \prod_{i=1}^\mathsf{M} A_i = \langle \mathcal{A}, \mathcal{D} \rangle.
\end{align}

As each element of $\vec{\Gamma}_i$ is of zero-mean and unit-variance, each element of $\mathcal{A}$ is of zero-mean and unit-variance.
Therefore, for the optimal choice of $\{d_j\}$, we have that,
\begin{align}
    {\tt LMSE} = \mathbb{E}\left[\left| \tilde{V} - \prod_{i=1}^\mathsf{M} A_i \right|^2 \right] = \| \mathcal{D} \|_F^2.
\end{align}
To derive a lower bound on ${\tt LMSE}$, it suffices to establish a lower bound on $\| \mathcal{D} \|_F^2$.

We first state the following two lemmas, originally established in \cite{cadambe2023differentially}. For completeness, we prove Lemmas~\ref{lemma::bound_snr} and~\ref{lemma::nullify} in Appendix~\ref{sec::proof_lemma}. 

\begin{lemma} \label{lemma::bound_snr}
For any set $\mathcal{S} \subseteq [\mathsf{N}]$ where $1 \le |\mathcal{S}|\le \mathsf{T}$, and for any constant $\bar{c}_j$ with $j\in\mathcal{S}$, the following inequality holds for all $i \in [\mathsf{M}]$,
\begin{align}
\left\|  \sum_{j\in \mathcal{S}} \bar{c}_j  \vec{w}_i^{(j)}   - \begin{bmatrix}\sqrt{\eta} \\ 0 \\ 0 \\ \vdots \\ 0\end{bmatrix} \right\|_2^2 \geq \frac{\eta}{1+{\tt SNR}^{*}(\epsilon)}. 
\end{align}

\end{lemma}

\begin{lemma} 
\label{lemma::nullify}
    For all $i\in [\mathsf{M}]$ and any set of nodes $\mathcal{S}$ with $1 \le |\mathcal{S}| \le \mathsf{T}$, there exists a vector 
    \begin{align}
        \vec{\alpha}_i = \begin{bmatrix}
            \alpha_{i}[1] \\
            \alpha_{i}[2] \\
            \vdots \\
            \alpha_{i}[\mathsf{N}+1]
        \end{bmatrix},
    \end{align}
    such that for node $j\in \mathcal{S}$,
    \begin{align}
        \vec{\alpha}_i^T \vec{w}_i^{(j)} = 0. \label{equ::nullify}
    \end{align}
    And 
    \begin{align}
        \frac{(\alpha_{i}[1])^2}{\| \vec{\alpha}_i \|_2^2} \ge \frac{1}{1+{\tt SNR}^*(\epsilon)}. \label{equ::alpha0}
    \end{align} 
\end{lemma}

By the system condition $\mathsf{M} \le \mathsf{N} \le \mathsf{M}\mathsf{T}$, there exist sets $\mathcal{S}_1, \mathcal{S}_2, \dots, \mathcal{S}_{\mathsf{M}}$ such that 
$\bigcup_{i=1}^{\mathsf{M}} \mathcal{S}_i = [\mathsf{N}]$, 
$\mathcal{S}_i \cap \mathcal{S}_j = \emptyset$ for any $i \neq j$, 
and $1 \le |\mathcal{S}_i| \le \mathsf{T}$ for all $i \in [\mathsf{M}]$.
Based on Lemma \ref{lemma::nullify}, for each set $\mathcal{S}_i$, there exists a vector $\vec{\alpha}_i$ that satisfies the properties stated in Lemma \ref{lemma::nullify}.

Because of \eqref{equ::nullify}, contracting the order-$\mathsf{M}$ tensor $\mathcal{D}$ along modes $2$ through $\mathsf{M}$ with the vectors $\vec{\alpha}_2,\ldots,\vec{\alpha}_{\mathsf{M}}$ yields 

\begin{align}
    \mathcal{D} \times_{\mathsf{M}} \vec{\alpha}_\mathsf{M} \times_{\mathsf{M}-1} \vec{\alpha}_{\mathsf{M}-1} \cdots  \times_2 \vec{\alpha}_2 = \sum_{j \in \mathcal{S}_1} c_j\vec{w}_1^{(j)} - \begin{bmatrix}
        (\eta)^{\frac{\mathsf{M}}{2}} \prod_{i=2}^\mathsf{M} \alpha_{i}[1] \\
        0 \\
        \vdots \\
        0
    \end{bmatrix} \in \mathbb{R}^{\mathsf{N}+1}
\end{align}
where $\times_k$ denotes the mode-$k$ tensor--vector product and $c_j \in \mathbb{R}$ are real coefficients induced by the contraction.

Based on Lemma \ref{lemma::lmse} and the fact that $1 \le |\mathcal{S}_1|\le \mathsf{T}$, we have that
\begin{align}
    \left\|\sum_{j \in \mathcal{S}_1 } c_j\vec{w}_1^{(j)} - \begin{bmatrix}
        (\eta)^{\frac{\mathsf{1}}{2}}\\
        0 \\
        \vdots \\
        0
    \end{bmatrix}\right\|^2_2 \ge \frac{\eta}{1+{\tt SNR}^*(\epsilon)}.
\end{align}

By performing multiplication on both sides, it follows that
\begin{align}
    \left\| \mathcal{D} \times_{\mathsf{M}} \vec{\alpha}_\mathsf{M} \times_{\mathsf{M}-1} \vec{\alpha}_{\mathsf{M}-1} \cdots  \times_2 \vec{\alpha}_2 \right\|_2^2 =
    \left\|\sum_{j \in \mathcal{S}_1} c_j\vec{w}_1^{(j)} - \begin{bmatrix}
        (\eta)^{\frac{\mathsf{M}}{2} \prod_{i=2}^\mathsf{M} \alpha_{i}[1] }\\
        0 \\
        \vdots \\
        0
    \end{bmatrix}\right\|^2_2 \ge \frac{\eta^{\mathsf{M}}\prod_{i=2}^\mathsf{M} (\alpha_{i}[1])^2 }{1+{\tt SNR}^*(\epsilon)}.
\end{align}

Recall that we aim to get a lower bound of $\| \mathcal{D} \|_F^2$. 
Note that $\| \mathcal{D} \|_F^2 \ge \| \mathcal{D} \|_2^2$ for any tensor, and
\begin{align}
    \left\| \mathcal{D} \times_{\mathsf{M}} \vec{\alpha}_\mathsf{M} \times_{\mathsf{M}-1} \vec{\alpha}_{\mathsf{M}-1} \cdots  \times_2 \vec{\alpha}_2 \right\|_2^2 \le \| \mathcal{D} \|_2^2 \| \vec{\alpha}_\mathsf{M}\|^2_2 \| \vec{\alpha}_{\mathsf{M}-1}\|^2_2 \cdots \| \vec{\alpha}_2\|^2_2.
\end{align}

Hence we have
\begin{align}
    \| \mathcal{D} \|_F^2 &\ge \| \mathcal{D} \|_2^2  \\
    & \ge \frac{ \left\| \mathcal{D} \times_{\mathsf{M}} \vec{\alpha}_\mathsf{M} \times_{\mathsf{M}-1} \vec{\alpha}_{\mathsf{M}-1} \cdots  \times_2 \vec{\alpha}_2 \right\|_2^2}{\| \vec{\alpha}_\mathsf{M}\|^2_2 \| \vec{\alpha}_{\mathsf{M}-1}\|^2_2 \cdots \| \vec{\alpha}_2\|^2_2}  \\
    &\ge \frac{\eta^{\mathsf{M}}\prod_{i=2}^\mathsf{M} (\alpha_i[1])^2 }{1+{\tt SNR}^*(\epsilon)} \frac{1}{\| \vec{\alpha}_\mathsf{M}\|^2_2 \| \vec{\alpha}_{\mathsf{M}-1}\|^2_2 \cdots \| \vec{\alpha}_2\|^2_2}  \\
    &= \frac{\eta^{\mathsf{M}}}{1+{\tt SNR}^*(\epsilon)} \prod_{i=2}^\mathsf{M} \frac{(\alpha_i[1])^2}{\|\vec{\alpha}_i\|_2^2}   \\
    &\overset{(a)}{\ge} \frac{\eta^\mathsf{M}}{\left(1+ {\tt SNR}^*(\epsilon)\right)^\mathsf{M}},
\end{align}
where $(a)$ is due to \eqref{equ::alpha0}.

Recall that ${\tt LMSE} =  \| \mathcal{D} \|_F^2$ for the optimal $\{d_j\}$. It then follows that
\begin{align}
    {\tt LMSE}\geq \frac{\eta^\mathsf{M}}{\left(1+ {\tt SNR}^*(\epsilon)\right)^\mathsf{M}}.
\end{align}


\section{Conclusion}

Secure multi-party computation has traditionally been studied in the cryptographic setting, where privacy guarantees are perfect --- colluding parties learn nothing beyond the function output --- but this perfection comes at a steep cost: either a large number of nodes or multiple rounds of interaction. In this paper, we adopt an \emph{approximation-theoretic} mindset, working directly over the reals and designing differentially private mechanisms and characterizing their accuracy guarantees. The result is a framework in which complex nonlinear computations can, in principle, be carried out in a single round of communication with a small number of nodes, in contrast to the demands of perfectly secure protocols.

Within this framework, we consider the canonical task of computing the product of $\mathsf{M}$ private inputs, extending prior results that were limited to pairwise products \cite{cadambe2023differentially, devulapalli2022differentially, hu2025differentially}. For the regime $(\mathsf{M}-1)\mathsf{T}+1 \le \mathsf{N} \le \mathsf{M}\mathsf{T}$, we characterize the optimal privacy--accuracy trade-off: the achievable ${\tt LMSE}$ scales as the $\mathsf{M}$-th power of the single-variable estimation error, and we provide a geometric interpretation of the developed coding-based privacy mechanism. For the minimal-redundancy regime $\mathsf{N}=\mathsf{T}+1$, we derive constructive achievability results and information-theoretic impossibility bounds that are tight in the high-privacy limit. Closing the remaining gap is an interesting direction of future work. In both regimes, by allowing a calibrated $\epsilon$-DP leakage, degree-$\mathsf{M}$ multiplications can be performed in a single round with as few as $(\mathsf{M}-1)\mathsf{T}+1$ workers, compared to $\mathsf{M}\mathsf{T}+1$ workers for classical one-round perfectly secure protocols or $O(\mathsf{M})$ rounds of interaction with $\mathsf{N}\ge 2\mathsf{T}+1$ workers. 

We view the results of this paper as a first step towards at least two broader research directions. The first is the development of a {theory of privacy-accuracy trade-offs for one-shot distributed computation}: the multiplication task studied here is a canonical building block, and a natural next step is to extend the framework to richer function classes --- including polynomial evaluations, multivariate statistics, and nonlinear operations arising in distributed machine learning. The second direction is a {research program that maps these information-theoretic insights into practical MPC protocols}. Our current analysis operates over the reals and relies on asymptotic scaling; translating these results to finite-precision arithmetic, finite block lengths, and practical noise distributions raises important questions about numerical stability and quantization effects. Addressing them is essential for bridging the gap between the theoretical privacy-accuracy trade-offs and practical secure computation systems.

\newpage 
\appendices

\section{Achievability Proof ($\mathsf{N}= \mathsf{T} + 1, \mathsf{N} < \mathsf{M}$)}
\label{sec::achievability_lessN}

This appendix analyzes the privacy--accuracy tradeoff for the case $\mathsf{N} = \mathsf{T} + 1$. Hence, we do not repeat the privacy analysis here.
The coding scheme used here is the same as in Section~\ref{sec::coding_schemes} (encoding polynomials~\eqref{equ::encoding_polynomial} and product polynomials \eqref{equ::general_product_polynomial_T=1} and  ~\eqref{equ::general_product_polynomial}), and the privacy analysis of Section~\ref{sec::dp_analysis} applies unchanged. 
To complete the proof of Theorem \ref{thm::achievability_lessN}, it suffices to show that ${\tt LMSE}$ satisfies the upper bound given in~\eqref{equ::achievability_lessN}.

For the accuracy analysis, by Lemma~\ref{lemma::poly_recovery} and the conditions in~\eqref{eq:limits}, $\mathsf{T}+1$ evaluations of $p(x)$ (specified in \eqref{equ::general_product_polynomial_T=1} and  ~\eqref{equ::general_product_polynomial}) suffice to asymptotically recover
\begin{align}
    c_0 = \prod_{i=1}^{\mathsf{M}} (A_i+R_i), \quad
    c_\mathsf{T} = \zeta_1(n) \sum_{i=1}^{\mathsf{M}}R_i \left( \prod_{l \neq i} (A_l+R_l) \right) + O(\zeta_2(n)^2).
\end{align}
$c_0$ and $c_{\mathsf{T}}$ 
as $n\to\infty$, since all coefficients of degree $j>\mathsf{T}$ satisfy $c_j/c_i\to 0$ for $i\le\mathsf{T}$. 
Here, we aim to estimate the desired product using 
\begin{align}
c_0 = \prod_{i=1}^{\mathsf{M}} (A_i+R_i), \quad 
c_0 + c_\mathsf{T}  = \prod_{i=1}^{\mathsf{M}} (A_i+ (1+\zeta_1(n)) R_i) + O(\zeta_2(n)^2).    
\end{align}

By Lemma~\ref{lemma::lmse}, give the observations $c_0, c_0 + c_\mathsf{T}$,  the ${\tt LMSE}$ is upper bounded by $\eta^\mathsf{M}/(1+{\tt SNR}_a)$, where $1+{\tt SNR}_a$ is as follows.

\begin{align}
    1+{\tt SNR}_a
    &= \frac
    {\left | \begin{matrix}
    (\eta + \sigma^2)^{\mathsf{M}} & (\eta + (1+\zeta_2(n))\sigma^2)^{\mathsf{M}} + O(\zeta_2(n)^4) \\
    (\eta + (1+\zeta_2(n))\sigma^2)^{\mathsf{M}} + O(\zeta_2(n)^4) & (\eta + (1+\zeta_2(n))^2\sigma^2)^{\mathsf{M}} + O(\zeta_2(n)^4)
    \end{matrix} \right | }
    {\left | \begin{matrix}
    (\eta + \sigma^2)^{\mathsf{M}} - \eta^\mathsf{M} & (\eta + (1+\zeta_2(n))\sigma^2)^{\mathsf{M}} - \eta^\mathsf{M} + O(\zeta_2(n)^4) \\
    (\eta + (1+\zeta_2(n))\sigma^2)^{\mathsf{M}} - \eta^\mathsf{M} + O(\zeta_2(n)^4) & (\eta + (1+\zeta_2(n))^2\sigma^2)^{\mathsf{M}} - \eta^\mathsf{M} + O(\zeta_2(n)^4)
    \end{matrix} \right | } \\
    & = \frac{(1+\frac{\eta}{\sigma^2})^\mathsf{M}}{(1+\frac{\eta}{\sigma^2})^\mathsf{M} - \left( \frac{\eta}{\sigma^2} \right)^\mathsf{M} - \mathsf{M} \left( \frac{\eta}{\sigma^2} \right)^{\mathsf{M}-1}} + O(\zeta_2(n)),
\end{align}
The achievable result in Theorem \ref{thm::achievability_lessN} is proved by substituting $\sigma^2 = \sigma^*(\epsilon)^2+\gamma'$ with sufficient small $\gamma'$ as shown in Section \ref{sec::dp_analysis}, i.e., 
\begin{equation}
   {\tt LMSE} (\mathcal{C}) \leq {\eta^\mathsf{M}} \frac{\sum_{k=0}^{\mathsf{M}-2} \binom{\mathsf{M}}{k} \left({\tt SNR}^*(\epsilon)\right)^k}{\left(1 + {\tt SNR}^*(\epsilon)\right)^\mathsf{M}}+\xi = {\eta^\mathsf{M}} \frac{(1+{\tt SNR}^*(\epsilon))^{\mathsf{M}} - \mathsf{M} {\tt SNR}^*(\epsilon)^{\mathsf{M}-1} - {\tt SNR}^*(\epsilon)^{\mathsf{M}}}{\left(1 + {\tt SNR}^*(\epsilon)\right)^\mathsf{M}}+\xi.
\end{equation}
for any $\xi>0$. 

\section{Converse Proof ($\mathsf{N}= \mathsf{T} + 1, \mathsf{N} < \mathsf{M}$)}
\label{sec::converse_lessN}
Consider the case $\mathsf{N}= \mathsf{T} + 1, \mathsf{N} < \mathsf{M}$.
In this subsection, we consider that $\mathbb{E}[A_i^2] = \eta$, and the result can be extended to the case $\mathbb{E}[A_i^2] < \eta$ based on Lemma \ref{lemma::lmse}.
Following the similar procedure in Section \ref{sec::converse}, there exist uncorrelated, zero-mean, unit-variance random variables $\frac{A_i}{\sqrt{\eta} }, \bar{R}_i^{(1)}, \cdots, \bar{R}_i^{(\mathsf{N})}$ for each $i\in [\mathsf{M}]$ \footnote{The set of random variables can be derived by multiplying the square root of the inverse of the covariance matrix.}.
In this case, let
\begin{align}
    \vec{\Gamma}_i = 
    \begin{bmatrix}
    \frac{A_i}{\sqrt{\eta} } & \bar{R}_i^{(1)} & \cdots & \bar{R}_i^{(\mathsf{N})}    
    \end{bmatrix}^T\in\mathbb{R}^{\mathsf{N}+1}.
\end{align}

Node $j$ stores the noisy version of $\frac{A_i}{ \sqrt{\eta}}$ as follows.
\begin{align}
    \tilde{A}_i^{(j)} =( \vec{w}_i^{(j)})^T \vec{\Gamma}_i,
\end{align}
where $\vec{w}_i^{(j)} \in \mathbb{R}^{\mathsf{N}+1}$ is the corresponding coefficient vector.

Similarly to Section \ref{sec::converse}, the estimated error can be written as a tensor inner product
\begin{align}
    \tilde{V} - \prod_{i=1}^\mathsf{M} A_i &= \left(\sum_{j=1}^\mathsf{N} d_j \left( \prod_{i\in [\mathsf{M}]} ( \vec{w}_i^{(j)})^T \vec{\Gamma}_i \right) \right) -  \prod_{i=1}^\mathsf{M} A_i = \langle \mathcal{A}, \mathcal{D} \rangle,
\end{align}
where tensors $\mathcal{A}$ and $\mathcal{D}$ are specified in Section \ref{sec::converse}.
Since each element of $\vec{\Gamma}_i$ has zero mean and unit variance, it follows that each element of $\mathcal{A}$ also has zero mean and unit variance. Hence, for the optimal choice of ${d_j}$, we have
\begin{align}
    {\tt LMSE} = \mathbb{E}\left[\left| \tilde{V} - \prod_{i=1}^\mathsf{M} A_i \right|^2 \right] = \| \mathcal{D} \|_F^2.
\end{align}
To get a lower bound of ${\tt LMSE}$, we aim to get a lower bound of $\| \mathcal{D} \|_F^2$ in the following.

Based on Lemma \ref{lemma::nullify}, for sets $\mathcal{S}_1=\{1\}, \mathcal{S}_2=\{2\}, \dots, \mathcal{S}_\mathsf{T}=\{\mathsf{T}\}$, there exist vectors $\vec{\alpha}_1, \vec{\alpha}_2, \dots, \vec{\alpha}_{\mathsf{T}}$ that satisfies the properties stated in Lemma \ref{lemma::nullify}.

Note that 
\begin{align}
    \mathcal{D} \times_{\mathsf{T}} \vec{\alpha}_{\mathsf{T}} \times_{\mathsf{T}-1} \vec{\alpha}_{\mathsf{T}-1} \cdots  \times_1 \vec{\alpha}_1 =  c_{\mathsf{N}}\vec{w}_{\mathsf{T}+1}^{(\mathsf{N})} \otimes \vec{w}_{\mathsf{T}+2}^{(\mathsf{N})} \otimes \cdots \otimes \vec{w}_{\mathsf{M}}^{(\mathsf{N})}
    - \eta^{\mathsf{T}/2}\prod_{i=1}^{\mathsf{T}} \alpha_{i}[1](\underbrace{\vec{e} \otimes \cdots \otimes\vec{e}}_{\mathsf{M} -\mathsf{T} \text{ times}})
    \in  \mathbb{R}^{\underbrace{\scriptstyle (\mathsf{N}+1) \times \cdots \times (\mathsf{N}+1)}_{\mathsf{M} -\mathsf{T} \text{ times}}},
\end{align}
where $c_{\mathsf{N}}$ is the constant derived via the linear combination.

Based on Lemma \ref{lemma::lmse}, to estimate $A_{\mathsf{T}+1} \cdot A_{\mathsf{T}+2} \cdot \cdots \cdot A_{\mathsf{M}}$ from only node ${\mathsf{N}}$, the achievable signal-noise ratio is 
\begin{align}
    {\tt SNR}'(\epsilon) = \frac{\eta^{\mathsf{M}-\mathsf{T}}}{(\eta+\sigma^*(\epsilon)^2)^{\mathsf{M}-\mathsf{T}} - \eta^{\mathsf{M}-\mathsf{T}}} =\frac{{\tt SNR}^*(\epsilon)^{\mathsf{M}-\mathsf{T}}}{(1+{\tt SNR}^*(\epsilon))^{\mathsf{M}-\mathsf{T}}-{\tt SNR}^*(\epsilon)^{\mathsf{M}-\mathsf{T}}},
\end{align}
Hence by Lemma \ref{lemma::lmse}, it follows that 
\begin{align}
    \left\|c_{\mathsf{N}}\vec{w}_{\mathsf{T}+1}^{(\mathsf{N})} \otimes \vec{w}_{\mathsf{T}+2}^{(\mathsf{N})} \otimes \cdots \otimes \vec{w}_{\mathsf{M}}^{(\mathsf{N})}
    - (\underbrace{\vec{e} \otimes \cdots \otimes\vec{e}}_{\mathsf{M} -\mathsf{T} \text{ times}})\right\|^2_F \ge \frac{\eta^{{\mathsf{\mathsf{M}-\mathsf{T}}}}}{1+{\tt SNR}'(\epsilon)}.
\end{align}

By performing multiplication on both sides, it follows that
\begin{align}
    \left\| \mathcal{D} \times_{\mathsf{T}} \vec{\alpha}_{\mathsf{T}} \times_{\mathsf{T}-1} \vec{\alpha}_{\mathsf{T}-1} \cdots  \times_1 \vec{\alpha}_1  \right\|_F^2 
    = \left \| c_{\mathsf{N}}\vec{w}_{\mathsf{T}+1}^{(\mathsf{N})} \otimes \vec{w}_{\mathsf{T}+2}^{(\mathsf{N})} \otimes \cdots \otimes \vec{w}_{\mathsf{M}}^{(\mathsf{N})}
    - \eta^{\mathsf{T}/2}\prod_{i=1}^{\mathsf{T}} \alpha_{i}[1](\underbrace{\vec{e} \otimes \cdots \otimes\vec{e}}_{\mathsf{M} -\mathsf{T} \text{ times}}) \right\|_F^2  \ge \frac{\eta^\mathsf{M} \prod_{i=1}^{\mathsf{T}} \alpha_{i}[1]^2}{1+ {\tt SNR}'(\epsilon)}.
\end{align}

Recall that we aim to get a lower bound of $\| \mathcal{D} \|_F^2$. 
Note the fact that
\begin{align}
    \left\| \mathcal{D} \times_{\mathsf{T}} \vec{\alpha}_{\mathsf{T}} \times_{\mathsf{T}-1} \vec{\alpha}_{\mathsf{T}-1} \cdots  \times_1 \vec{\alpha}_1  \right\|_F^2  \le \| \mathcal{D} \|_F^2 \|\vec{\alpha}_\mathsf{T}\|^2_2 \| \vec{\alpha}_{\mathsf{T}-1}\|^2_2 \cdots \| \vec{\alpha}_1\|^2_2.
\end{align}

Hence we have
\begin{align}
    \| \mathcal{D} \|_F^2 &\ge 
    \frac{\left\| \mathcal{D} \times_{\mathsf{T}} \vec{\alpha}_{\mathsf{T}} \times_{\mathsf{T}-1} \vec{\alpha}_{\mathsf{T}-1} \cdots  \times_1 \vec{\alpha}_1  \right\|_F^2}{\|\vec{\alpha}_\mathsf{T}\|^2_2 \| \vec{\alpha}_{\mathsf{T}-1}\|^2_2 \cdots \| \vec{\alpha}_1\|^2_2}
     \\
    & \ge \frac{\eta^\mathsf{M} }{1+ {\tt SNR}'(\epsilon)} \prod_{i=1}^{\mathsf{T}} \frac{\alpha_{i}[1]^2}{\|\vec{\alpha}_i\|_2^2}  \\
    &\overset{(a)}{\ge} \eta^\mathsf{M} \frac{1}{(1+{\tt SNR}^*(\epsilon))^\mathsf{T}} \frac{1}{1+ {\tt SNR}'(\epsilon)} \\
    & = \eta^\mathsf{M} \frac{\sum_{k=0}^{\mathsf{M}-\mathsf{T}-1}
\binom{\mathsf{M}-\mathsf{T}}{k}
\left({\tt SNR}^*(\epsilon)\right)^k}{(1+{\tt SNR}^*(\epsilon))^\mathsf{M}},
\end{align}
where $(a)$ is due to \eqref{equ::alpha0}.

Therefore,
\begin{equation} 
    {\tt LMSE} (\mathcal{C}) \geq \eta^\mathsf{M} \frac{\sum_{k=0}^{\mathsf{M}-\mathsf{T}-1} \binom{\mathsf{M}-\mathsf{T}}{k} \left({\tt SNR}^*(\epsilon)\right)^k}{(1+{\tt SNR}^*(\epsilon))^\mathsf{M}} = \eta^\mathsf{M} \frac{(1+{\tt SNR}^*(\epsilon))^{\mathsf{M}-\mathsf{T}} - {\tt SNR}^*(\epsilon)^{\mathsf{M}-\mathsf{T}}}{(1+{\tt SNR}^*(\epsilon))^\mathsf{M}}.
\end{equation}

\section{Proofs of Technical Lemmas}
\label{sec::proof_lemma}

\subsection{Proof of Lemma \ref{lemma::existence_Z}}
\label{sec::proof_existence_Z}

For a given DP parameter $\epsilon$, define 
$$\sigma^2 = \sigma^*(\epsilon)^2+\gamma',$$ 
where $\gamma'>0$.
For a fixed variance level $\sigma^2$, we define,
\begin{align}
    \epsilon^* = \inf_{N, \mathbb{E}[N^2] \ge \sigma^2} \sup_{\mathcal{B} \subseteq \mathcal{B}(\mathbb{R}), B_0, B_1 \in  \mathbb{R}, |B_0-B_1|\le 1} \ln \left( \frac{\mathbb{P}(B_0+ N\in \mathcal{B})}{\mathbb{P}(B_1+ N\in \mathcal{B})}\right),
\end{align}
where $N \in \mathbb{R}$ is a zero-mean random variable, and $\mathcal{B}(\mathbb{R})$ denote the Borel $\sigma$-algebra on $\mathbb{R}$. 
Note that the noise variance $\mathbb{E}[N^2]$ is strictly larger than $\sigma^*(\epsilon)^2$. 
Since $\sigma^*(\epsilon)$ strictly decreases in the DP parameter $\epsilon$ (according to \eqref{equ::optimalsigma}) and $\mathbb{E}[N^2] > \sigma^*(\epsilon)^2$, it follows that 
$\epsilon^* < \epsilon$.
Consequently, for a DP parameter $\bar{\epsilon}$ with $\epsilon^* < \bar{\epsilon} < \epsilon$, there exists a random noise variable $N^*$ such that $\mathbb{E}[(N^*)^2]\le \sigma^2$ satisfying:
\begin{align}
    \sup_{\mathcal{B}\subseteq \mathcal{B}(\mathbb{R}), -1<\lambda<1} \frac{\mathbb{P}(A_1+N^* \in \mathcal{B})}{\mathbb{P}(A_1+N^*+\lambda \in \mathcal{B})} \le e^{\bar{\epsilon}} \le e^{\epsilon}.
\end{align}

\subsection{Proof of Lemma \ref{lemma::_z_z'}}
\label{sec::proof_z_z'}

Since the matrix $\bar{\mathbf{G}}$ has full rank $\mathsf{T}$ by construction, the colluding nodes can get, via a one-to-one map of $\vec{Z}$, the following quantities:
$
    (A_1 +R_1)\bar{\mathbf{G}}^{-1}\vec{1}+ 
    \begin{bmatrix}
        \zeta_1(n) R_1 & \zeta_2(n) S_{1,1} & \cdots & \zeta_2(n)  S_{1, \mathsf{T}-1}
    \end{bmatrix}^T.
$

We now show that $\vec{g'}_1^T\vec{1}\neq 0$.
As $\bar{\mathbf{G}}^{-1}\bar{\mathbf{G}}=\mathbf{I}$, we have that
$
\vec{g'}_1^T \begin{bmatrix}
        h_1 & h_2 & \cdots & h_\mathsf{T}
    \end{bmatrix}^T = 1
$
and
$
\vec{g'}_1^T
\begin{bmatrix}
    \vec{g}_1 & \vec{g}_2 & \cdots & \vec{g}_\mathsf{T}
\end{bmatrix}^T = \vec{0}^T.
$
The first equation shows that $\vec{g'}_1^T$ is not an all-zero row vector.
According to the coding scheme, the matrix $
\begin{bmatrix}
    1 & 1 & \cdots & 1 \\
    \vec{g}_1 & \vec{g}_2 & \cdots & \vec{g}_\mathsf{T}
\end{bmatrix}^T$ is full-rank, together with $
\vec{g'}_1^T
\begin{bmatrix}
    \vec{g}_1 & \vec{g}_2 & \cdots & \vec{g}_\mathsf{T}
\end{bmatrix}^T = \vec{0}^T
$,  $\vec{g'}_1^T\vec{1}$ cannot be zero.

We can then normalize the first component of the mapped $\vec{Z}$ and obtain $A_1 + \left(1 + \frac{1}{ \vec{g'}_1^T\vec{1}} \zeta_1(n)\right) R_1$.
This quantity can subsequently be used to eliminate the corresponding $R_1$ terms from the remaining components of $\vec{Z}$.\footnote{Note that we only consider the non-trivial case where $\vec{g'}_j^T\vec{1}\neq 0$ with $j\in \{2, \cdots, T\}$. If $\vec{g'}_j^T\vec{1} = 0$, privacy is well-preserved as only noise remains.} 
Hence, we can obtain the vector $\vec{Z}'$ shown in Lemma \ref{lemma::_z_z'}.

\subsection{Proof of Lemma \ref{lemma::bound_snr}}
\label{sec::proof_bound_snr}

For any set $\mathcal{S} \subseteq [\mathsf{N}]$ with $1 \le |\mathcal{S}|\le \mathsf{T}$, let
\begin{align}
    A_i^{\mathcal{S}} &= \sum_{j\in \mathcal{S}} \bar{c}_j \tilde{A}_i^{(j)}  \\
    &= \sum_{j\in \mathcal{S}} \bar{c}_j\begin{bmatrix}
    \frac{A_i}{\sqrt{\eta} } & \bar{R}_i^{(1)} & \cdots & \bar{R}_i^{(\mathsf{N})}    
    \end{bmatrix}\vec{w}_i^{(j)}  \\
    &= \frac{\sum_{j\in\mathcal{S}} \bar{c}_j \vec{w}_i^{(j)}[1]}{\sqrt{\eta}} A_i
    + \sum_{k=2}^{\mathsf{N}+1}  \left( \sum_{j\in\mathcal{S}} \bar{c}_j \vec{w}_i^{(j)}[k] \right) \bar{R}_i^{(k-1)}  \\
    &= \frac{\sum_{j\in\mathcal{S}} \bar{c}_j \vec{w}_i^{(j)}[1]}{\sqrt{\eta}}\left( A_i + \tilde{R}_i^\mathcal{S}\right),
\end{align}
where $\vec{w}_i^{(j)}[k]$ denotes the element of $\vec{w}_i^{(j)}$ with index $k$,
$\tilde{R}_i^\mathcal{S} =  \frac{\sqrt{\eta}}{\sum_{j\in\mathcal{S}} \bar{c}_j \vec{w}_i^{(j)}[1]} \sum_{k=2}^{\mathsf{N}+1}  \left( \sum_{j\in\mathcal{S}} \bar{c}_j \vec{w}_i^{(j)}[k] \right) \bar{R}_i^{(k-1)}$.
According to the definition of $\mathsf{T}$-node $\epsilon$-DP in Definition~\ref{def::dp}, the collective information by any subset of at most $\mathsf{T}$ colluding nodes must satisfy $\epsilon$-DP.
Since $A_i^{\mathcal{S}} = \frac{\sum_{j\in\mathcal{S}} \bar{c}_j \vec{w}_i^{(j)}[1]}{\sqrt{\eta}}\left( A_i + \tilde{R}_i^\mathcal{S}\right)$ is a linear combination of the information available to the subset $\mathcal{S}$, the post-processing property of DP \cite{dwork2014algorithmic} implies that $A_i + \tilde{R}_i^\mathcal{S}$ must also satisfy $\epsilon$-DP.
Consequently, the effective noise $\tilde{R}_i^\mathcal{S}$ must satisfy
\begin{align}
    \mathbb{E}\left[\left(\tilde{R}_i^\mathcal{S}\right)^2\right] \ge \sigma^*(\epsilon)^2,   
\end{align}
where $\sigma^*(\epsilon)^2$ denotes the minimal noise variance required to ensure $\epsilon$-DP, as characterized in Lemma~\ref{lemma::pramod_results}.
Consequently, the signal-to-noise ratio of $A_i^{\mathcal{S}}$ as an estimator of $A_i$ is upper bounded by ${\tt SNR}^*(\epsilon)=\frac{\eta}{\sigma^*(\epsilon)^2}$.

Based on Lemma \ref{lemma::lmse}, we have
\begin{align}
    \left\|  \sum_{j\in \mathcal{S}} \bar{c}_j  \vec{w}_i^{(j)}   - \begin{bmatrix}\sqrt{\eta} \\ 0 \\ 0 \\ \vdots \\ 0\end{bmatrix} \right\|^2 
    \overset{(a)}{=} 
    \mathbb{E} \left[\left| \begin{bmatrix}
    \frac{A_i}{\sqrt{\eta} } & \bar{R}_i^{(1)} & \cdots & \bar{R}_i^{(\mathsf{N})}    
    \end{bmatrix}  \left(\sum_{j\in \mathcal{S}} \bar{c}_j  \vec{w}_i^{(j)}   - \begin{bmatrix}\sqrt{\eta} \\ 0 \\ 0 \\ \vdots \\ 0\end{bmatrix} \right) \right|^2 \right]
    \geq \frac{\eta}{1+{\tt SNR}^{*}(\epsilon)},
\end{align}
where $(a)$ follows from the fact that random variables $\frac{A_i}{\sqrt{\eta} }, \bar{R}_i^{(1)}, \cdots, \bar{R}_i^{(\mathsf{N})}$ are uncorrelated, zero-mean, unit-variance.

\subsection{Proof of Lemma \ref{lemma::nullify}}
\label{sec::proof_nullify}

For any vector $\vec{w}_i = \begin{bmatrix}
        w_{i}[1] &
        \cdots &
        w_{i}[\mathsf{N}+1]
\end{bmatrix}^T \in \mathbb{R}^{\mathsf{N}+1}$ in the span of $\{\vec{w}_i^{(j)}: j\in \mathcal{S}\}$.
Let $\mathcal{S}=\{s_1, s_2, \dots, s_t\}$, and assume that $\vec{w}_i=\sum_{j=1}^t \beta_j \vec{w}_i^{(s_j)}$.
Based on Lemma \ref{lemma::bound_snr}, it follows that, for any constant $\gamma$,
\begin{align}
    \mathbb{E}\left[\left|\gamma (\beta_1 \tilde{A}_i^{(s_1)} + \beta_2 \tilde{A}_i^{(s_2)}+ \cdots \beta_t \tilde{A}_i^{(s_t)})- \frac{A_i}{\sqrt{\eta}}\right|^2\right] 
    =\mathbb{E}\left[\left|\gamma \begin{bmatrix}
        \frac{A_i}{\sqrt{\eta}} & \bar{R}_i^{(1)} & \cdots & \bar{R}_i^{(\mathsf{N})}
    \end{bmatrix} \vec{w}_i- \frac{A_i}{\sqrt{\eta}}\right|^2\right]
    \ge \frac{1}{1+{\tt SNR}^*(\epsilon)}.
\end{align}

As each element of $ \begin{bmatrix}
        \frac{A_i}{\sqrt{\eta}} & \bar{R}_i^{(1)} & \cdots & \bar{R}_i^{(\mathsf{N})}
    \end{bmatrix}$ is zero-mean and unit-variance, it follows that
\begin{align}
    (\gamma {w}_{i}[1] -1)^2 + \sum_{k=2}^{\mathsf{N}+1} \gamma^2 (w_{i}[k])^2 \ge \frac{1}{1+{\tt SNR}^*(\epsilon)}.
\end{align}
Let $\gamma = \frac{w_{i}[1]}{\sum_{k=1}^{\mathsf{N}+1} (w_{i}[k])^2}$, we have that
\begin{align}
    \frac{(w_{i}[1])^2}{\sum_{k=2}^{\mathsf{N}+1} (w_{i}[k])^2} \le {\tt SNR}^*(\epsilon).
\end{align}

As $|\mathcal{S}| \le \mathsf{T}$, the null space of $\{\vec{w}_i^{(j)}: j\in \mathcal{S}\}$ is non-trivial.
$\begin{bmatrix}
    1 &
    0 &
    \cdots &
    0
\end{bmatrix}^T \in \mathbb{R}^{\mathsf{N}+1}$ cannot lie in the span of $\{\vec{w}_i^{(j)}: j\in \mathcal{S}\}$, otherwise ${\tt SNR}^*(\epsilon) = \infty$.
Hence, we only consider the case where 
$\begin{bmatrix}
    1 &
    0 &
    \cdots &
    0
\end{bmatrix}^T \in \mathbb{R}^{\mathsf{N}+1}$ does not lie in the span of $\{\vec{w}_i^{(j)}: j\in \mathcal{S}\}$.
By the rank-nullity theorem, there always exists a vector $\vec{w}_i = \begin{bmatrix}
        w_i[1] &
        w_i[2] &
        \cdots &
        w_i[\mathsf{N}+1]
\end{bmatrix}^T$ in the span of $\{\vec{w}_i^{(j)}: j\in \mathcal{S}\}$, and $\vec{\alpha}_i=\begin{bmatrix}
            \alpha_{i}[1] &
            \alpha_{i}[2] &
            \cdots &
            \alpha_{i}[\mathsf{N}+1]
        \end{bmatrix}^T$ that in the null space of $\{\vec{w}_i^{(j)}: j\in \mathcal{S}\}$, such that
\begin{align}
    \vec{w}_i+ \vec{\alpha}_i = \begin{bmatrix}
        1 \\ 0 \\\vdots \\0
    \end{bmatrix}.
\end{align}

As $\vec{\alpha}_i^T \vec{w}_i = 0$, we have
\begin{align}
    w_{i}[1] \alpha_{i}[1] = -\sum_{k=2}^{\mathsf{N}+1} w_{i}[k] \alpha_{i}[k] = \sum_{k=2}^{\mathsf{N}+1} (w_{i}[k])^2 = \sum_{k=2}^{\mathsf{N}+1} (\alpha_{i}[k])^2.
\end{align}

Hence we have
\begin{align}
    \frac{\|\vec{\alpha}_i \|_2^2}{(\alpha_{i}[1])^2} = 1+ \frac{\sum_{k=2}^{\mathsf{N}+1} (\alpha_{i}[k])^2}{(\alpha_{i}[1])^2} = 1+ \frac{\sum_{k=2}^{\mathsf{N}+1} (w_{i}[k])^2}{(\frac{\sum_{k=2}^{\mathsf{N}+1} (w_{i}[k])^2}{w_{i}[1]})^2} =
    1 + \frac{(w_{i}[1])^2}{\sum_{k=2}^{\mathsf{N}+1} (w_{i}[k])^2} \le 1+ {\tt SNR}^*(\epsilon).
\end{align}

Hence, the lemma is proved.

\section{Proof of Proposition~\ref{prop::CkDk}}
\label{sec::proof_CkDk}

To prove Proposition~\ref{prop::CkDk}, we proceed in two parts.

Part~(i):
Suppose $\mathsf{N} = (\mathsf{M}-1)\mathsf{T}+1$. For each $k=0,1,\ldots,\mathsf{M}-1$, define
\begin{align}\label{equ::Ck}
    C_k = \sum_{\mathcal{S} \subseteq [\mathsf{M}],\, |\mathcal{S}|=k}
    \left( \prod_{i\in \mathcal{S}} R_i \right)
    \left( \prod_{l \notin \mathcal{S}} (A_l+R_l) \right).
\end{align}
We first show that there exist functions $g_k$ such that
\[
    \lim_{n\to\infty}
    g_k\!\bigl(\tilde{V}^{(1)},\ldots,\tilde{V}^{(\mathsf{N})}\bigr)
    = C_k,
\]
where the coding scheme (and hence $\tilde{V}^{(j)}$) depends on $n$ through the parameters $\zeta_1(n),\zeta_2(n)$ defined in Section~\ref{sec::coding_schemes}.

Part~(ii):
We then show that $D_k$ is a linear combination of $C_0, C_1, \ldots, C_k$. Consequently, $D_0, D_1, \ldots, D_{\mathsf{M}-1}$ are uniquely determined by $C_0, C_1, \ldots, C_{\mathsf{M}-1}$.

\medskip
Combining Parts (i) and (ii), it follows that $\{D_k\}_{k=0}^{\mathsf{M}-1}$ can be asymptotically recovered from the local computation results $\{\tilde{V}^{(j)}\}_{j\in[\mathsf{N}]}$, with error vanishing as $n\to\infty$.

\subsection*{Part~(i): Asymptotic recovery of $C_k$ from $\{\tilde{V}^{(j)}\}$}

We first state the following auxiliary lemma.

\begin{lemma}[Asymptotic Polynomial Coefficient Recovery]\label{lemma::poly_recovery}
Let $p(x) = \sum_{k=0}^{d} c_k x^k$ be evaluated at $r$ distinct points $x_1, \ldots, x_r$ with $r \le d$. If $\lim_{n\to\infty} c_j / c_i = 0$ for all $i \le r-1 < j \le d$ (with $c_i \ne 0$), then there exist functions $f_k$ such that
\[
    \lim_{n\to\infty} f_k\!\bigl(p(x_1),\ldots,p(x_r)\bigr) = c_k
    \quad\text{for each } k=0,\ldots,r-1.
\]
\end{lemma}

\begin{proof}
The evaluations satisfy $[p(x_1),\ldots,p(x_r)]^T = \mathbf{V}[c_0,\ldots,c_{r-1}]^T + \boldsymbol{\delta}$, where $\mathbf{V}$ is the invertible Vandermonde matrix and $\boldsymbol{\delta}$ collects contributions from higher-degree terms. Since $\|\boldsymbol{\delta}\|/\|[c_0,\ldots,c_{r-1}]^T\| \to 0$ as $n\to\infty$, applying $\mathbf{V}^{-1}$ recovers $c_0,\ldots,c_{r-1}$ with error vanishing as $n\to\infty$.
\end{proof}

For the case $\mathsf{T}=1$, we use the encoding polynomial defined in \eqref{equ::encoding_polynomial_t=1} to construct the overall product polynomial
\begin{align}\label{equ::general_product_polynomial_T=1}
    p(x) = \prod_{i=1}^{\mathsf{M}} p_i(x) =\prod_{i=1}^{\mathsf{M}} \left( (A_i + R_i) + \zeta_1(n) R_i x \right) = \sum_{k=0}^{\mathsf{M}-1}  \zeta_1(n)^k C_k x^k + O\left( \zeta_1(n)^\mathsf{M} \right).
\end{align}
As $n \to \infty$, the remainder term $O\!\left( \zeta_1(n)^{\mathsf{M}} \right)$ becomes negligible relative to the leading $\mathsf{M}$ terms. Hence, by Lemma~\ref{lemma::poly_recovery}, $\mathsf{N}=\mathsf{M}$ evaluations of $p(x)$ suffice to recover $C_0, C_1, \dots, C_{\mathsf{M}-1}$.

For the case $\mathsf{T} \ge 2$, using the encoding polynomial defined in \eqref{equ::encoding_polynomial}, we construct
\begin{align}\label{equ::general_product_polynomial}
    p(x) = \prod_{i=1}^{\mathsf{M}} p_i(x) = \prod_{i=1}^\mathsf{M} \left((A_i + R_i) + \zeta_2(n)\sum_{t=1}^{\mathsf{T}-1} S_{i,t} x^t + \zeta_1(n) R_i x^\mathsf{T} \right) = \sum_{k=0}^{\mathsf{M} \mathsf{T}} c_k x^k,
\end{align}
where
\begin{align}
\label{equ::product_polynomial_coefficients}
c_k
=
\sum_{{
t_1+\cdots+t_{\mathsf{M}} = k,\;
t_i \in \{0,\dots,\mathsf{T}\}
}}
\prod_{i=1}^{\mathsf{M}}
\begin{cases}
A_i + R_i, & t_i=0, \\[4pt]
\zeta_2(n)\, S_{i,t_i}, & 1 \le t_i \le \mathsf{T}-1, \\[4pt]
\zeta_1(n)\, R_i, & t_i=\mathsf{T}.
\end{cases}
\end{align}

\begin{lemma}
    Assume \eqref{eq:limits} holds, i.e.,
    $$\lim_{n \to \infty} \frac{\zeta_1{(n)}}{\zeta_2{(n)}} = \lim_{n \to \infty} \zeta_2{(n)} = \lim_{n\to\infty} \frac{\zeta_2{(n)}^{\mathsf{T}/(\mathsf{T}-1)}}{\zeta_1{(n)}} = 0.$$
    Then, for each $\ell = 0,1,\dots,\mathsf{M}-1$, the coefficient $c_{\ell\mathsf{T}}$ satisfies
    \begin{align}
        c_{\ell\mathsf{T}} = \zeta_1(n)^\ell C_\ell + o\big(\zeta_1(n)^\ell\big).
    \end{align}
\end{lemma}

\begin{proof}
Fix $\ell\in\{0,\ldots,\mathsf{M}-1\}$. By~\eqref{equ::product_polynomial_coefficients}, each term in $c_{\ell\mathsf{T}}$ is indexed by a tuple $(t_1,\ldots,t_\mathsf{M})\in\{0,\ldots,\mathsf{T}\}^\mathsf{M}$ satisfying $\sum_i t_i = \ell\mathsf{T}$, and contributes a monomial of the form $\zeta_1(n)^k \zeta_2(n)^s$ (up to bounded factors), where $k = |\{i : t_i = \mathsf{T}\}|$ and $s = |\{i : 1 \le t_i \le \mathsf{T}-1\}|$. To show the theorem, it suffices to show that, for any $k \neq 0$,  $\zeta_1(n)^k \zeta_2(n)^s = o(\zeta_n^{\ell}).$ 

Since $\sum_{i:\, 1\le t_i\le \mathsf{T}-1} t_i = (\ell - k)\mathsf{T}$ and each such summand $t_i$ is at most $\mathsf{T}-1$, we have $s(\mathsf{T}-1) \ge (\ell-k)\mathsf{T}$, hence $s \ge \frac{(\ell-k)\mathsf{T}}{\mathsf{T}-1}$. Therefore,
$$\zeta_1(n)^k \zeta_2(n)^s = O\!\left(\zeta_1(n)^k\,\zeta_2(n)^{(\ell-k)\mathsf{T}/(\mathsf{T}-1)}\right).$$
For $k < \ell$, it follows from~\eqref{eq:limits} that
\begin{align}
    \lim_{n\to \infty} \frac{\zeta_1(n)^k\,\zeta_2(n)^{(\ell-k)\mathsf{T}/(\mathsf{T}-1)}}{\zeta_1(n)^\ell}
    = \lim_{n\to \infty} \left(\frac{\zeta_2(n)^{\mathsf{T}/(\mathsf{T}-1)}}{\zeta_1(n)}\right)^{\ell-k} = 0,
\end{align}
so every such term with $k \neq 0$ is $o(\zeta_1(n)^\ell)$. The unique tuple with $k = \ell$ contributes exactly $\zeta_1(n)^\ell C_\ell$. Hence $c_{\ell\mathsf{T}} = \zeta_1(n)^\ell C_\ell + o(\zeta_1(n)^\ell)$.
\end{proof}


From \eqref{equ::Ck}, recovery of $c_{\ell\mathsf{T}}$ (up to the known scale $\zeta_1(n)^\ell$) suffices to recover $C_\ell$ with error vanishing as $n\to\infty$. Moreover, for any indices $i,j$ satisfying $0 \le i \le (\mathsf{M}-1)\mathsf{T} < j \le \mathsf{M}\mathsf{T}$, the conditions in~\eqref{eq:limits} imply $\lim_{n\to\infty} c_j/c_i = 0$; every term in $c_j$ with $j>(\mathsf{M}-1)\mathsf{T}$ carries at least one additional factor of $\zeta_2(n)$ or $\zeta_1(n)$ relative to $c_i$ with $i\le(\mathsf{M}-1)\mathsf{T}$.
By Lemma~\ref{lemma::poly_recovery}, $(\mathsf{M}-1)\mathsf{T}+1$ evaluations of $p(x)$ suffice to asymptotically recover $C_0, C_1, \ldots, C_{\mathsf{M}-1}$. This establishes Part~(i).

\subsection*{Part~(ii): $D_k$ as a linear combination of $C_0,\ldots,C_k$}

The proof proceeds by recursion on $k$.

\textit{Base case ($k=0$).} Since $A_i+ Z_i = \alpha (A_i + R_i)$,
\begin{align}
    D_0 = \prod_{i=1}^{\mathsf{M}} (A_i + Z_i) = \alpha^{\mathsf{M}} \prod_{i=1}^{\mathsf{M}} (A_i + R_i) = \alpha^{\mathsf{M}} C_0.
\end{align}

\textit{Recursive step.} Suppose $D_0, D_1, \dots, D_{k-1}$ can each be written in terms of $C_0, C_1, \dots, C_{k-1}$ for some $k-1 \ge 0$. Expanding $C_k$:
\begin{align}
    C_k &=  \sum_{\mathcal{S} \subseteq [\mathsf{M}], |\mathcal{S}|=k} \left(\prod_{i\in \mathcal{S}} R_i \right) \left( \prod_{l \notin \mathcal{S}} (A_l+R_l) \right)  \\
    &\overset{(a)}{=} \left( \frac{1}{\alpha} \right)^{\mathsf{M}-k}  \sum_{\mathcal{S} \subseteq [\mathsf{M}], |\mathcal{S}|=k} \left(\prod_{i\in \mathcal{S}} \left( \frac{A_i+ Z_i}{\alpha} - A_i\right) \right) \left( \prod_{l \notin \mathcal{S}} (A_l+Z_l) \right),
\end{align}
where $(a)$ uses $A_i+Z_i = \alpha(A_i+R_i)$ and $R_i = \frac{A_i+Z_i}{\alpha}-A_i$. Expanding and regrouping by $|\mathcal{U}|=u$ (fixing $\mathcal{U}\subseteq\mathcal{S}$ and summing over all supersets $\mathcal{S}$) gives
\begin{align} \label{equ::Ck_Dk}
    C_k = \sum_{u=0}^{k} C_{k,u}\,D_u,
    \qquad C_{k,u}=(-1)^u \,\alpha^{u-\mathsf{M}}\binom{\mathsf{M}-u}{k-u}.
\end{align}
Since the diagonal coefficient $C_{k,k}=(-1)^k \alpha^{k-\mathsf{M}}\neq 0$, $D_k$ can be obtained as
\begin{align}
    D_k = \frac{1}{C_{k,k}}\left(C_k-\sum_{u=0}^{k-1} C_{k,u}\,D_u\right).
\end{align}
The recursion gives $D_k$ as a linear combination of $C_0,\ldots,C_k$. Combining both parts, $\{D_k\}_{k=0}^{\mathsf{M}-1}$ are asymptotically recovered from $\{\tilde{V}^{(j)}\}_{j\in[\mathsf{N}]}$.

\newpage
\bibliographystyle{IEEEtran}
\bibliography{main}

\end{document}